\documentclass[letterpaper,11pt]{article}
\newif\ifdraft
\drafttrue

\newif\iflong
\longtrue
\usepackage[utf8]{inputenc}
\usepackage{microtype}
\usepackage[letterpaper,margin=1in]{geometry}
\usepackage[numbers,sort&compress]{natbib}
\usepackage{xcolor,xspace}
\usepackage{amsmath}
\usepackage{amssymb}

\usepackage{amsthm}
\usepackage{paralist}
\usepackage{enumitem}
\setitemize{noitemsep}

\usepackage{textgreek}
\usepackage{graphicx}
\usepackage{framed}
\usepackage{wrapfig}

\usepackage[small]{caption}
\usepackage{booktabs}
\usepackage[unicode]{hyperref}
\usepackage[capitalize, nameinlink]{cleveref}
\usepackage{doi}
\usepackage{thmtools,thm-restate}
\usepackage{xspace}
\usepackage{nicefrac}
\usepackage[noend]{algorithmic}
\usepackage{algorithm}

\usepackage{datetime2} %% Just to see when the version was made

\theoremstyle{plain}
\newtheorem{theorem}{Theorem}[section]
\newtheorem{lemma}[theorem]{Lemma}

\newtheorem{corollary}[theorem]{Corollary}
\newtheorem{claim}[theorem]{Claim}

\newtheorem*{conjecture*}{Conjecture}
\newtheorem{observation}[theorem]{Observation}

\theoremstyle{definition}
\newtheorem{definition}[theorem]{Definition}

\theoremstyle{remark}
\newtheorem{remark}[theorem]{Remark}
\newtheorem*{remark*}{Remark}

\usepackage{todonotes}

\newcommand{\trycolor}{\textsc{TryColor}}
\newcommand{\slackgeneration}{\textsc{SlackGeneration}}

\newcommand{\CONGEST}{\ensuremath{\mathsf{CONGEST}}\xspace}
\newcommand{\LOCAL}{\ensuremath{\mathsf{LOCAL}}\xspace}

\newcommand{\eps}{\varepsilon}
\newcommand{\poly}{\operatorname{\text{{\rm poly}}}}

\newcommand{\logstar}[1]{\log^{*} #1}
\DeclareMathOperator{\E}{\mathbb{E}}

\newcommand{\myparagraph}[1]{\vspace{-0.5cm}\paragraph{#1}}

\newcommand{\lovasz}{Lov\'{a}sz\xspace}

\newcommand{\Exp}{\mathbb{E}}

\renewcommand{\phi}{\varphi}

\newcommand{\cC}{\mathcal{C}}

\newcommand{\cE}{\mathcal{E}}

\newcommand{\cG}{\mathcal{G}}

\newcommand{\cN}{\mathcal{N}}

\newcommand{\cO}{\mathcal{O}}
\newcommand{\cP}{\mathcal{P}}
\newcommand{\cR}{\mathcal{R}}

\newcommand{\cV}{\mathcal{V}}

\newenvironment{mycover}
{\list{}{\listparindent 0pt
        \itemindent    \listparindent
        \leftmargin    1cm
        \rightmargin   1cm
        \parsep        0pt}%
    \raggedright
    \item\relax}
{\endlist}

\newcommand{\myaff}[1]{\,$\cdot$\, {\small #1}\par\smallskip}

\hypersetup{
    colorlinks=true,
    linkcolor=black,
    citecolor=black,
    filecolor=black,
    urlcolor=[HTML]{0088CC},
    pdftitle={Distributed Graph Coloring Made Easy},
    pdfauthor={Yannic Maus}
}

\begin{document}

\title{Fast Distributed Brooks' Theorem}
%\title{Fast and Efficient Distributed $\Delta$-Coloring}

\begin{mycover}
    {\huge\bfseries Fast  Distributed Brooks' Theorem\par}
%    {\huge\bfseries Bandwidth Efficient  Delta Coloring  \par}
    \bigskip
    % \today // 
%    \DTMnow %% \DTMcurrenttime
    \bigskip
  
    \textbf{Manuela Fischer}
	\myaff{Reykjavik University, Iceland}
	
%	\myaff{Reykjav\'{i}k University, Iceland}
	\textbf{Magn\'{u}s M. Halld\'{o}rsson}
	\myaff{Reykjavik University, Iceland}
	
	\textbf{Yannic Maus}
	\myaff{TU Graz, Austria}
	
\end{mycover}
\medskip

\thispagestyle{empty}
\begin{abstract}
We give a randomized $\Delta$-coloring algorithm in the \LOCAL model that runs in $\poly \log \log n$ rounds, where $n$ is the number of nodes of the input graph and $\Delta$ is its maximum degree. 
This means that randomized $\Delta$-coloring is a rare distributed coloring problem with an upper and lower bound in the same ballpark, $\poly\log\log n$, given the known $\Omega(\log_\Delta\log n)$ lower bound [Brandt et al., STOC '16].

Our main technical contribution is a constant time reduction to a constant number of $(deg+1)$-list coloring instances, for $\Delta = \omega(\log^4 n)$, resulting in a $\poly \log\log n$-round \CONGEST algorithm for such graphs.
This reduction is of independent interest for other settings, including providing a new proof of Brooks' theorem for high degree graphs, and leading to a constant-round Congested Clique algorithm in such graphs.

When $\Delta=\omega(\log^{21} n)$, our algorithm even runs in $O(\log^* n)$ rounds, showing that the base in the $\Omega(\log_\Delta\log n)$ lower bound is unavoidable. 

Previously, the best \LOCAL algorithm for all considered settings used a logarithmic number of rounds. Our result is the first \CONGEST algorithm for $\Delta$-coloring non-constant degree graphs. 

\clearpage

\tableofcontents
\thispagestyle{empty}
\end{abstract}

%\tableofcontents

\clearpage

\title{Fast and Efficient Distributed Brooks' Theorem}

\section{Introduction}

\setcounter{page}{1}

In the \emph{graph coloring problem}, we are to color the nodes of a graph so that adjacent nodes receive different colors.
In the distributed setting, the graph to be colored is the same as the communication network.
The objective is both to minimize the number of colors and the time complexity of the algorithm.
Graph coloring is fundamental to distributed computing, as a way of \emph{breaking symmetry}.
In fact, it was the topic of the paper introducing the \LOCAL model \cite{linial92}.

In local models of distributed computing, a communication network is abstracted as a graph $G=(V,E)$ with $n$ nodes and maximum degree $\Delta$. Communication happens in  synchronous rounds. In each round, each node can send a message through each of its incident edges to its neighbors. At the end of the computation, each node should output its own part of the solution, e.g., its own color.  
In the \CONGEST model each message is limited to $O(\log n)$-bits, while they can be of unbounded size in the \LOCAL model.

Distributed algorithmics is typically concerned with problems that are easy to solve greedily in a centralized setting. For coloring, this corresponds to using $\Delta+1$ colors, which has the key property that any partial coloring can be extended to a valid full solution.
$O(\log n)$-round randomized algorithms were known for $\Delta+1$-coloring early on \cite{luby86,alon86,linial92,johansson99},
while recent progress has led to $\poly\log\log n$-round algorithms in the \LOCAL model \cite{BEPSv3,HSS18,CLP20,RG20} and the \CONGEST model \cite{HKMT21}. 
The $\Delta$-coloring problem is of a 
different nature. Despite only a single color of difference, it is non-trivial to solve centrally and even the existence of $\Delta$-colorings is a classic result in graph theory:
\begin{quote} 
	\textbf{Brooks' theorem} \cite{brooks_1941}: Any (connected) graph that is neither an odd cycle nor a clique on $\Delta+1$ nodes can be $\Delta$-colored.
\end{quote}
Understanding $\Delta$-coloring and the related but simpler problem of \emph{sinkless orientation} has received enormous attention since the publication of  an $\Omega(\log_{\Delta}\log n)$-round ($\Omega(\log_{\Delta}n)$-round) randomized (deterministic) lower bound  for both problems \cite{brandt2016LLL,chang2016exponential}. In fact, approaches in understanding the $\Delta$-coloring problem have ushered in the technique of \emph{round elimination} \cite{Brandt19,brandt2016LLL} that has resulted in numerous new lower bounds for various related problems \cite{Balliu2019,BrandtO20,BBORules,BBKOmis,BBKO22}. 
However, the $\Delta$-coloring problem has so far withstood all attempts in improving either the lower bound or the algorithms. This even caused authors of a recent paper to name it  \emph{'Distributed $\Delta$-coloring plays hide-and-seek'} \cite{BBKO22}. 

Finding or ruling out a lower bound for the $\Delta+1$-coloring problem is considered one of the most central open questions in distributed algorithms \cite{linial92,barenboimelkin_book,disc16_coloring}. Understanding the complexity of the $\Delta$-coloring problem may help shed a light on this \cite{BBKO22}. Let us detail on how the state-of-the-art complexities for both problems differ. 
On constant-degree graphs, $\Delta+1$-coloring has a $O(\log^* n)$-round algorithm \cite{linial92}, even a deterministic one, while randomized algorithms for $\Delta$-coloring provably require $\Omega(\log_{\Delta}\log n)$ rounds. On general graphs, the gaps between known upper and lower bounds are much greater, but again the state-of-the-art round complexities (in \LOCAL) are $O(\log^3 \log n)$ for $\Delta+1$-coloring \cite{CLP20} and only $O(\log\Delta)+\poly\log\log n=O(\log n)$ for $\Delta$-coloring \cite{GHKM21}.
One might be tempted to conjecture that $\Delta$-coloring is much harder and requires entirely different algorithmic techniques.

\subsection{Results}
In this paper, we investigate this question. As our main result we provide an algorithm for $\Delta$-coloring graphs whose maximum degree is at least polylogarithmic in the number of nodes. We obtain the following theorem.
\begin{restatable}{theorem}{tmMainLogstar}
	\label{thm:mainLogstar}
	There is a randomized $\poly\log\log n$-round \CONGEST algorithm, that w.h.p. $\Delta$-colors any graph with $\Delta= \omega(\log^3 n)$ that does not contain a $\Delta+1$-clique.  If $\Delta=\omega(\log^{21}n)$, it runs in $O(\logstar n)$ rounds. 
\end{restatable}
As the aforementioned \LOCAL algorithm from \cite{GHKM21} with runtime $O(\log \Delta)+\poly\log\log n$ yields an algorithm with runtime $\poly\log\log n$ whenever the maximum degree is at most polylogarithmic, we obtain the following corollary. 
\begin{restatable}{corollary}{thmPolylog}
	\label{corr:mainPolylog}
	There is a randomized $\poly\log\log n$-round \LOCAL algorithm that w.h.p. $\Delta$-colors any graph that does not contain a $\Delta+1$-clique, for any $\Delta\geq 3$.
\end{restatable}
Hence, on general graphs, our algorithm for $\Delta$-coloring matches the best known runtime bound for $\Delta+1$-coloring. We believe that the most interesting aspect of \Cref{thm:mainLogstar} is its runtime on graphs with large maximum degree. Here, we improve the state-of-the-art from $O(\log n)$ rounds to $O(\logstar n)$ rounds, matching the fastest algorithm for $\Delta+1$-coloring such graphs \cite{HKNT22}. Thus, the complete known picture for the two problems is that their complexities are provably separated on small degree graphs, both problems can be solved in $\poly\log\log n$ for intermediate values of $\Delta$, and as soon as $\Delta$ is large enough both can be solved in $O(\logstar n)$ rounds. 

The latter shows that the $\Delta$ appearing in the base of the logarithm in the runtime lower bound of $\Omega(\log_\Delta \log n)$ rounds for randomized $\Delta$-coloring algorithms \cite{brandt2016LLL} is not an artifact of the technique used, as the problem can indeed be solved faster for larger degrees. Note that as soon as the maximum degree is at least polylogarithmic, the lower bound actually vanishes. Similarly, it is not known whether the celebrated $\Omega(\logstar n)$ lower bound by Linial \cite{linial92} (randomized by Naor \cite{naor95}) for $\Delta+1$-coloring path and ring graphs ($\Delta=2$) extends to larger degrees. So, while both $\Delta+1$- and $\Delta$-coloring admit now extremely fast algorithms for large degree graphs, it is not known whether one can improve on these results.

From a technical and conceptual point of view, the main strength of \Cref{thm:mainLogstar} lies in the algorithm itself. In fact, we show the following:

\begin{framed}
    \centering
   For graphs with a sufficiently large polylogarithmic maximum degree, there is a constant time reduction of the $\Delta$-coloring problem to a constant number of $(deg+1)$-list coloring instances.  
\end{framed}
The objective in the \emph{$(deg+1)$-list coloring problem} is to compute a valid coloring where each node outputs a color from its list of allowed colors; the list is part of the node's input and its size exceeds the node's degree. The reduction is helpful for several reasons: (1) Any $(deg+1)$-list coloring instance is solvable.  (2) Any partial solution can be extended to a full solution, and hence the problem is conceptually simpler than the $\Delta$-coloring problem.  (3) The problem has been extensively studied in the distributed setting, e.g., \cite{FHK,MT20,HKMT21}, culminating in recent \LOCAL \cite{HKNT22} and \CONGEST \cite{HNT22} algorithms running in $\poly\log\log n$ rounds, that even run in $O(\logstar n)$ rounds if all lists are sufficiently large. Hence, the first part of \Cref{thm:mainLogstar} follows immediately from the reduction and the algorithm from \cite{HNT22}. (4) The reduction provides a new self-contained distributed proof for Brooks' theorem, at least for graphs with sufficiently large maximum degree. (5) The constant time reduction is deterministic, except for a single round in which some carefully selected nodes pick a random color and keep it if no neighbor wants to get colored with the same color. 
Hence, the reduction is robust enough that we expect it could help also in other models of computing.

We exemplify the latter by providing an algorithm in the Congested Clique. 
Similar to the \CONGEST model, in the Congested Clique model,  the nodes of the input graph can communicate with each other in synchronous rounds by sending $O(\log n)$-bit messages. However, instead of only communicating through the edges of the input graph they can communicate in all-to-all fashion, that is, each node can send one message containing $O(\log n)$ bits to each other node in the network per round.

\begin{restatable}{theorem}{thmMaincongestedClique}
	\label{cor:CCmain}
	There is a randomized $O(1)$-round Congested Clique algorithm that $\Delta$-colors w.h.p.\ any graph without a $\Delta+1$-clique, when $\Delta=\omega(\log^{4+\epsilon} n)$.
\end{restatable}
\Cref{cor:CCmain} is obtained by combining our reduction with a previous $\Delta+1$-coloring algorithm \cite{CFGUZ18}. Even though that algorithm is not designed for the list coloring problem, our reduction yields list coloring instances that are "simple" enough to be dealt with within their framework. 

Despite the attention that the $\Delta$-coloring problem has gotten, to the best of our knowledge, \Cref{cor:CCmain} yields the first $\Delta$-coloring algorithms in these models and \Cref{thm:mainLogstar} is also the  first \CONGEST algorithm for $\Delta$-coloring non-constant degree graphs.

Open Problem 8 in \cite{BBKO22} asks for a $\Delta$-coloring algorithm in the \LOCAL model that is genuinely different from the algorithms in \cite{PS95,GHKM21}. We solve this partially, in the context of non-constant degree graphs. Unlike  previous distributed $\Delta$-coloring algorithms \cite{panconesi95,GHKM21}, our results do not rely on any deep graph-theoretic insights such as \emph{degree choosability} and Gallai trees (see \cite{GHKM21} for details).

\subsection{\texorpdfstring{The Challenges of $\Delta$-coloring}{The Challenges of Delta-coloring}}
\label{ssec:challenges}
The $\Delta$-coloring problem clearly differs from $\Delta+1$-coloring only in that each node has one fewer color to use. This unit deficit makes all the difference though, since if we color the node's neighbors carelessly, we may no longer be able to color the node.  
To deal with this we need to ensure that each node has a unit of \emph{slack}, meaning excess of colors over the competition for them. Also note that in a solution to the problem every node with degree $\Delta$ has to have two neighbors with the same color. 
\iftrue
\begin{wrapfigure}{L}{0.65\textwidth}
	\includegraphics[width=0.65\textwidth, page=1]{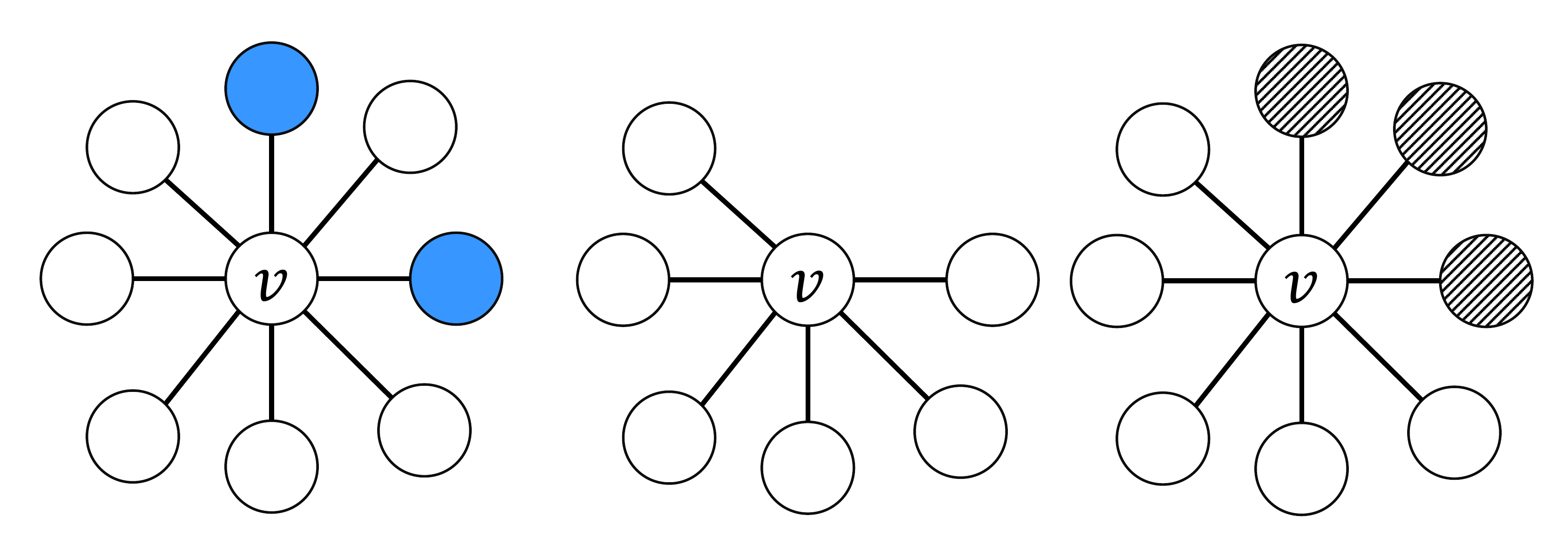}
	\caption{The three types of slack. An uncolored node can obtain permanent slack by two neighbors receiving the same color (left); a priori slack by having a degree $<\Delta$ (middle); or temporary slack by stalling some of its neighbors, to be colored later (right).}
	\label[figure]{fig:slack-types}
\end{wrapfigure}
\fi

When designing a coloring algorithm in which not all nodes are colored at the same time, there are different types of slack. 
A node gets slack if two of its neighbors are assigned the same color. It also gets a temporary reprieve if a neighbor is stalled to be colored later. 
We call the former \emph{permanent slack} and the latter \emph{temporary slack}. The easiest form is the \emph{a priori} slack that a node of low-degree has. These three types of slack are illustrated in \cref{fig:slack-types}.
Classic proofs of Brooks' theorem revolve around finding just one  non-adjacent pair of nodes with a common neighbor whose removal does not disconnect the graph.  Such a pair must exist in a $\Delta+1$-clique free graph.
The pair is then colored with the same color, giving its common neighbor permanent slack, and as long as we color it last, its neighbors have temporary slack, and if we stall their coloring, their neighbors will have temporary slack, etc., that is one can $\Delta$-color the whole graph with a layered approach.

The challenges towards a $O(\logstar n)$-round algorithm for $\Delta$-coloring become more evident when looking at the tools to obtain such runtimes for $(\Delta+1)$-coloring. 
The standard initial step is to partition the nodes into the sparse nodes and clusters of dense nodes known as \emph{almost-cliques (AC)}  \cite{HSS18,ACK19}, see \Cref{lem:acd} for the details.
Then, a standard step  in randomized coloring algorithms is to give slack to \emph{sparse} nodes -- those with many pairs of non-adjacent neighbors -- via the following one-round random step, which we call \emph{SlackGeneration}: 
Each node of a random subset chooses a color at random and keeps it if no neighbor also picked that color.
Thus, after a single round of color trials, all uncolored sparse nodes have permanent slack and induce a $(deg+1)$-list coloring instance.  Therefore,  \emph{dense} nodes are  the main challenge for $\Delta$-coloring. As we detail soon, dense parts of the graph impose strong requirements on which nodes need to be colored consistently. Thus, in addition to this \emph{probabilistically generated slack} for sparse nodes, we also require to handcraft slack deterministically. 

We illustrate some typical difficult cases, which also motivate our design choices. 
For the sake of this exposition,  it's best to think of an AC as a clique on approximately $\Delta$ nodes, possibly with a single edge missing, also see \Cref{fig:ac-types} on page \pageref{fig:ac-types}.
Then, in order to reduce to the $(deg+1)$-list coloring problem it suffices to give just one node of an AC slack, i.e., a \emph{toehold}. 
Namely, by deferring it, its neighbors get temporary slack, which can be recursively provided to the rest of the AC. Still, as we shall see, no single method suffices to deal with all ACs. 
Some ACs can obtain such a toehold probabilistically;
others can have it handcrafted; while yet others are best handled with temporary slack.
In the example of \Cref{fig:ac-types} (left) on page \pageref{fig:ac-types}, there is exactly one pair of non-adjacent nodes, and they \emph{must} be colored the same for the instance to be $\Delta$-colored. In the example on the right, we need to same-color an outside neighbor and an internal node, so that their common neighbor gets permanent slack. In this case, it is hard to do so in a deterministic distributed manner due to dependencies,
as the outside neighbors may themselves be in ACs of the same type.

However, it is not always possible to find a non-edge within the AC or use a probabilistic argument on the outside neighbors. For example, consider the ACs  in  \Cref{fig:guardedRunaway} on page \pageref{fig:guardedRunaway}. Here, all ACs have too few distinct outside neighbors to get probabilistic slack with high probability guarantee when coloring nodes randomly. 
For those, we either need to handcraft slack by same-coloring an inside-outside pair (\Cref{fig:guardedRunaway}, left, on page \pageref{fig:guardedRunaway}), or find an outside neighbor to be colored later (\Cref{fig:guardedRunaway}, right, on page \pageref{fig:guardedRunaway}). In the latter case, we need the later-colored neighbor to have permanent slack.

Even with a method to handle each type, the problem is fragile, as we must ensure that the different solutions are compatible. Any probabilistic slack generation needs to leave essential features sufficiently intact, and a solution for one type of AC cannot destroy the solution for another type, or one requires a completely new approach.

In this section we have only discussed the challenges for an efficient reduction to the $(deg+1)$-list coloring problem. In order to use such a reduction for an $O(\logstar n)$-round algorithm, these list coloring instances require additional properties, namely, the lists need to be large enough to leverage the list coloring result of \cite{HKNT22,HNT22}.

\subsection{Related Work}
\label{ssec:related}
The only works on $\Delta$-coloring general graphs in the \LOCAL model are a seminal paper from 1995 \cite{PS95} and a paper first published in 2018 \cite{GHKM21}. Both rely on so-called \emph{degree choosable subgraphs}, i.e., subgraphs with a certain structure such that they can be $deg$-list colored. In a nutshell, the algorithms in \cite{GHKM21} first remove all degree subgraphs that have at most a certain fixed radius $r$. As a result, the remaining graph has lots of additional tree-like structure that can be exploited. Degree choosable subgraphs can then be colored \emph{brute-force} 
in the end, since, by definition, any partial coloring outside the subgraph can be extended to a full solution on the subgraph. The algorithm of \cite{PS95} also relies on coloring such large subgraphs brute-force.  Hence, these algorithms are crucially based on learning the full topology in (often even non-constant radius)  neighborhoods, which is infeasible for an efficient \CONGEST implementation. 

Additional papers refer to $\Delta$-coloring, e.g., \cite{KuhnRecursive20,MT20,RG20,GK21}.
These do not address the $\Delta$-coloring problem directly, but 
provide improved $(deg+1)$-list coloring algorithms, which is used as a subroutine in \cite{GHKM21}.

The randomized $\Omega(\log_{\Delta}\log n)$ and deterministic $\Omega(\log_{\Delta} n)$ lower bounds for $\Delta$-coloring \cite{brandt2016LLL} already hold on tree topologies and these are known to be tight \cite{chang2016exponential}. 

There is only one result on $\Delta$-coloring in the \CONGEST model \cite{MU21}, that is, a $\poly\log \log n$-round randomized algorithm,  if $\Delta$ is constant. The result follows from a general statement about so called locally checkable labeling (LCL) problems \cite{naor95}. It is known that any LCL problem (on constant-degree graphs) with a sublogarithmic randomized \LOCAL algorithm also has a $\poly\log\log n$-round \LOCAL algorithm \cite{ChangHierarchy17,RG20}. In \cite{MU21} the authors show that the result also holds in the \CONGEST model. Hence, their $\Delta$-coloring algorithm is non-explicit and it is unlikely that the same technique extend to large values of $\Delta$. 

These general results also show that the $\Delta$-coloring problem is closely linked to the distributed \lovasz\ local lemma problem (LLL). Any problem with a sublogarithmic algorithm in the \LOCAL model on constant degree graphs can be solved in the same time as the LLL problem \cite{ChangHierarchy17}. $\Delta$-coloring is such a problem \cite{GHKM21}, and in fact on trees both problems can be solved in $\Theta(\log_{\Delta}\log n)$ rounds \cite{chang2016exponential} (and $\Delta$ coloring is in fact solved via an LLL-like approach). It is conjectured that the same runtimes hold on general graphs \cite{ChangHierarchy17}. 

A single-pass streaming algorithm was recently given for $\Delta$-coloring that uses $\tilde{O}(n)$ memory \cite{AKM22}. 
Although the setting is quite different from ours, many of the issues they deal with and the structures they propose are very similar. They give a similar classification of almost-cliques, with different treatment of the three different settings described in the previous subsection: a) ACs containing a non-edge (our "nice" ACs; their "holey" and "unholey critical"), b) ACs that benefit from slack generation (our "ordinary"; their "lonely small"), and c) ACs that need outside help (our "difficult"; their "friendly small"). 
Let us examine these in some more details.

The distinction between b) and c) depends on if there exists a node outside the clique that has many neighbors in the clique. These are 'special nodes' as we call them, or 'friends' as in \cite{AKM22}, are crucial to the arguments, yet have not appeared before in the literature. ACs without special nodes are loosely enough connected to receive slack in the initial random color trial step (for case b)). Our arguments in this case involves the somewhat heavy application of Talagrand's inequality; they can process the cliques individually and compute a coloring for each via a random graph argument similar to \cite{ACK19}. 
Treating the first type, a), is relatively easy distributively (our 'nice' ACs). That is harder in streaming since full access to the adjacency lists is not available, but made possible with the sparse recovery technique. 
Treating the third type is the hardest for the distributed solution, because of the possible dependencies. We split these 'difficult' ACs into two types, 'runaway' and 'guarded', corresponding to the two cases shown in \Cref{fig:guardedRunaway}. This case is also the last and most involved one in \cite{AKM22} but the fact that the cliques can be resolved sequentially in the streaming setting simplifies the arguments. Just as we do for 'guarded' ACs, they pair an outside node with a non-neighbor in the AC and assign them the same color. 
We ensure that the outside node participates only in this one pair (which necessitates a separate treatment of the other type, 'runaway' ACs) to avoid dependencies, while they recolor the relevant outside node if it was already colored in previous phases. 

At the end of the day, the challenging issues are different in streaming and distributed settings and the precise definitions and the solutions are somewhat different. Yet, the similarity between the classifications of ACs and 'special' nodes is illuminating, given that these concepts have not appeared before.

There is a vast amount of randomized and deterministic distributed vertex coloring algorithms 
that color with more than $\Delta$ colors. 
The excellent monograph \cite{barenboimelkin_book} covers many of these results until 2013. We only cover a few selected newer results, and, e.g., completely spare the large body of work on edge colorings or colorings of special graph classes. 
An ultrafast $O(\logstar n)$-round algorithm was already published in 2010, though using $O(\Delta^{1+\Omega(1)})$ colors \cite{SW10}.
In 2016, the aforementioned influential concept of almost-clique decompositions was introduced to the area and yielded the first sublogarithmic $\Delta+1$-coloring algorithm \cite{HSS18}.  In 2018, the first  $\poly\log\log n$-round algorithm appeared \cite{CLP20,RG20}. 
All these results hold in the \LOCAL model. Only in 2021, the first $\poly\log\log n$ algorithm in the \CONGEST model appeared \cite{HKMT21}. 
Very recently, the aforementioned results have been extended to the $(deg+1)$-list coloring problem in both the \LOCAL \cite{HKNT22} and \CONGEST \cite{HNT22} models,  answering an open question in \cite{CLP20}.

A recent breakthrough provided the first deterministic $\poly\log n$-round algorithms for $\Delta+1$-coloring  in the \LOCAL model \cite{RG20}; the algorithm is also used as a subroutine in many of the aforementioned results. Shortly after, these result were extended to the \CONGEST model \cite{BKM20,GK21}.
An orthogonal research direction aims for fast deterministic \CONGEST and \LOCAL algorithms mainly as a function of the maximum degree $\Delta$, that is, the dependence on $n$ is limited to the unavoidable  $O(\logstar n)$ term \cite{linial92}. There have been several recent algorithms \cite{Barenboim16,FHK,BEG18,MT20,M21} with the state-of-the art being $O(\sqrt{\Delta\log\Delta}+\logstar n)$ rounds for $(\deg+1)$-list coloring in \LOCAL and $O(\Delta^{3/4}+\logstar n)$ rounds in \CONGEST \cite{Barenboim16}. 

To the best of our knowledge, there are no known $\Delta$-coloring algorithms for the Congested Clique.
$O(1)$-round algorithms are known for $\Delta+1$-coloring, both randomized \cite{CFGUZ18} and deterministic \cite{CDP21}.

\subsection{Outline}
In \Cref{sec:acd}, we present notation and the almost-clique decomposition (modified to our needs). 
In \Cref{sec:algHighDegree}, we present the steps of our reduction, including reasons for their correctness (in a nutshell). In \Cref{sec:detailedReduction}, we formally prove their correctness, except for the analysis of the randomized slack generation. To ease presentation, the latter appears in \Cref{sec:slackGeneration}. In \Cref{sec:CC}, we reason how to implement our algorithm in the Congested Clique.

\section{Almost Clique Decompositions \& Notation}
\label{sec:acd}
Denote $[k]=\{0,\ldots,k-1\}$. Consider a graph $G=(V,E)$. For a node $v\in V$ let $N(v)$ be the set of its neighbors and for a subset $S\subseteq V$ let $N(S)=\bigcup_{v\in S}N(v)$. $deg(v)$ denotes the degree of a node and for $S\subseteq V$ we denote $N_S(v)=S\cap N(v)$ and $deg_S(v)=|N_S(v)|$. We use $g(\Delta) = \omega(f(\Delta))$ in a statement if the statement holds for a (sufficiently large) constant $C > 0$ such that 
$g(\Delta) \ge C f(\Delta)$.
In the \emph{$\ell$-list coloring problem} each node is given a list of allowed colors of size at least $\ell$.  The objective is to compute a valid coloring in which each node outputs a color from its list.

\subsection{Almost Clique Decomposition}

The \emph{(local) sparsity} $\zeta_v$ of a vertex $v$ is defined as $\zeta_v=\frac{1}{\Delta}\left(\binom{\Delta}{2}-m(N(v))\right)$, where for a set $X$ of vertices, $m(X)$ denotes the number of edges in the subgraph $G[X]$. Roughly speaking, $\zeta_v$ is proportional to the number of missing edges in the graph induced by $v$'s neighborhood. However, note that a node with degree $\ll \Delta$ also has large sparsity. 

Almost-clique decompositions have become a useful tool for distributed graph coloring, e.g., \cite{HSS18,ACK19,HKMT21}. Here, we use a form of the decomposition with slightly stronger properties. 
It partitions $V$ into $V_{sparse}$, the \emph{sparse} nodes, and $\cC$, the \emph{dense} nodes.
It further partitions the dense nodes into \emph{almost cliques (AC)} $C_1, C_2, \ldots, C_t$.
Let $C(v)$ denote the almost-clique that a given dense node $v$ belongs to.
The \emph{outside degree} $e(v)$ of a dense node $v$ is the number of neighbors of $v$ outside $C(v)$. 
\footnote{We use this term to distinguish from the \emph{external degree} used in \cite{HSS18,HKMT21} and elsewhere, denoting the number of neighbors in \emph{other} almost cliques (not counting sparse neighbors).} The \emph{anti-degree} $a(v)$ of a dense node $v$  is the number of non-neighbors of $v$ in $C(v)$.

\begin{restatable}[ACD computation]{lemma}{lemACDcomputation}
	\label{lem:acd}
	For any graph $G=(V,E)$, there is a partition (\emph{almost-clique decomposition (ACD)} of $V$ into sets $V_{sparse}$ and $C_1, C_2, \ldots, C_t$ such that:
	\begin{compactenum}
		\item Each node in $V_{sparse}$ is $\Omega(\epsilon^2\Delta)$-sparse\ ,
		\item For every  $i\in [k]$, $(1-\eps)\Delta\le |C_i|\le (1+3\eps)\Delta$\ ,
		\item Each dense node $v$ has at least $(1-4\eps)\Delta$ neighbors in $C(v)$:   $|N(v)\cap C(v)|\ge (1-4\eps)\Delta$\ ,
		\item For each node $u \not\in C_i$, $|N(u)\cap C_i|\le (1-2\eps)\Delta$\ .
	\end{compactenum}
	Further, there is an $O(1)$-round \CONGEST algorithm that for a given graph with maximum degree $\Delta$ and a given constant $0<\eps<1/20$ computes a valid ACD, w.h.p., if $\Delta= \omega(\log^2 n)$.
\end{restatable}
\begin{proof}[Proof of \Cref{lem:acd}]
	We first use the \CONGEST algorithm from \cite[arXiv, Lemma 9.1]{HKMT21} to compute an ACD $(V', D_1, D_2, \ldots, D_t)$ in $O(1)$ rounds.  The algorithm is stated for $\eps=1/3$ but it can be run in the same number of rounds and with the same guarantees for any given constant $\eps\leq 1/3$. 
	Then, we modify the decomposition as follows.
	Let $C_i$ consist of $D_i$ along with any node $u \in V'$ with at least $(1-4\eps)\Delta$ neighbors in $D_i$, for each $i \in [t]$.
	Let $V_{sparse} = V \setminus \cup_i C_i$. This completes the construction.
	
	First, observe that the decomposition is well-defined, as a node $u\in V'$ cannot have more than $(1-4\eps)\Delta$ neighbors in more than on $D_i$. 
	
	Second, observe that by \cite[arXiv, Lemma 6.2]{HKMT21} the nodes in $V'$ have sparsity $(\eta^2/4)\Delta$ where $\eta=\eps/108$. So, they are $\Omega(\eps^2\Delta)$-sparse. This holds then also for the nodes of $V_{sparse}$, establishing (1).
	
	Third, each node in $D_i$ has at least $(1-4\eps)\Delta$ neighbors in $D_i\subseteq C_i$, by the definition of an ACD in \cite[Lemma 9.1 + ACD Def]{HKMT21}, and the same holds by construction for the nodes of $C_i \setminus D_i$. This establishes (3).

	By  \cite[Lemma 9.1 + ACD Def]{HKMT21}, each $D_i$ satisfies $(1-\eps)\Delta \le |D_i| \le (1+\eps)\Delta$.
	We claim that $|C_i \setminus D_i| \le 2\epsilon$, which establishes (2).
	Each node in $D_i$ has at most $\eps \Delta$ outside neighbors, so the number of edges with exactly one endpoint in $D_i$ is at most $\eps\Delta|D_i| \le \eps(1+\eps)\Delta^2$.
	Each node of $C_i \setminus D_i$ is incident on at least $(1-4\eps)\Delta$ such edges.
	Thus, $|C_i \setminus D_i| \le \eps (1+\eps)\Delta/(1-4\eps) \le 2\eps\Delta$, for $\eps \le 1/20$. Hence, $|C_i|\le |D_i|+|C_i\setminus D_i| \le (1+3\epsilon)\Delta$, establishing (2).
	
	Finally, by definition, each node in $V_{sparse}$ has at most $(1-4\epsilon)\Delta$ neighbors in $D_i$. Also using $|C_i\setminus D_i|\leq 2\eps\Delta$, we obtain that each node in $V_{sparse}$ has  at most $(1-2\eps)\Delta$ neighbors in $C_i$. 
	By (3), each node in $C_j$, $j\ne i$, has at most $4\epsilon\Delta \le (1-2\epsilon)\Delta$ neighbors in $C_i$, establishing (4).
\end{proof}

The next observation is immediate from  properties of an ACD (\Cref{lem:acd} (2) and (3)).
\begin{restatable}{observation}{obsoutsideAndAntiDegree}
	\label{obs:outsideAndAntiDegree}
	For each node $u$ in an AC $C$ we have  $a(u) \le 7\epsilon\Delta$ and
	$e(u) \le 4 \epsilon \Delta$. 
\end{restatable}
\begin{proof}[Proof of \Cref{obs:outsideAndAntiDegree}]
The bound on $a(v)$ additionally requires that the AC is of size $(1+3\eps)\Delta$ by \Cref{lem:acd} (2). 
	The bound on $e(v)$ follows as $v$ has at least $(1-4\eps)\Delta$ neighbors in $C(v)$ by \Cref{lem:acd} (3). 
\end{proof}

\section{\texorpdfstring{Ultrafast $\Delta$-coloring}{Ultrafast Delta-coloring}}
\label{sec:algHighDegree}
The objective of this and the next section is to argue the following main theorem.
\tmMainLogstar*

In this section, we present the high-level structure of our reduction to a constant number of $deg+1$-list coloring instances, including the arguments for correctness in a nutshell. In  \Cref{sec:detailedReduction}, we present the corresponding formal proofs. We solve the instances via the following result.
\begin{lemma}[List coloring \cite{HKNT22,HNT22}]
	\label{lem:listColoring}
	There is a randomized \CONGEST algorithm to $(deg+1)$-list-color any graph in $O(\log^3 \log n)$ rounds, w.h.p. If all lists are of size at least $\omega(\log^7 n)$, then the algorithm runs in $O(\logstar n)$ rounds. 
\end{lemma}

\subsection{Definition and Almost-Clique Classification}
\label{ssec:ACpartition}
To avoid an unsolvable situation where a partial coloring cannot be extended to a full $\Delta$-coloring, the order in which we color nodes is crucial. We now define the notation needed to present this order. We partition the ACs of a given ACD into three types: ordinary, nice, and difficult, with the last further divided into guarded and runaway ones.  We next explain the intuition behind this partition, formal definitions follow hereafter. Ordinary ACs can profit from an initial randomized coloring step. These ACs have many outgoing edges connecting the AC with many  distinct external neighbors, and hence coloring a constant fraction of these external neighbors and a constant number of the nodes inside the AC randomly will with high probability produce an uncolored toehold node inside the clique;  two of its neighbors,  one external and one within the AC, will have chosen the same color. Difficult ACs are such that they do have many outgoing edges but many of them are incident to the same external (special) node.  Thus, a similar reasoning based on coloring nodes randomly is doomed to fail as its success depends too much on the randomness of that special node and hence difficult ACs are treated as follows. We create slack differently, depending on whether the special node is shared by more than one AC (runaway ACs) or not (guarded ACs). For guarded ACs we same-color the special node (a protector) and one of its non-neighbors inside the clique. This procedure itself is a $(deg+1)$-list coloring problem on a virtual graph where the nodes are pairs of vertices. Intuitively, we could proceed similarly for runaway ACs, but since different ACs share the same special node it is difficult to perform this \emph{same-coloring procedure} on a complicated virtual graph. Instead, we use the special node itself as the toehold (an escape) by deferring coloring it. It can be colored at the end as it can obtain slack through coloring the nodes in its adjacent runaway ACs randomly. Lastly, the intuition for nice ACs is yet a different one. One example of a nice AC is a $\Delta+1$ clique with a single edge missing. Here, we need to same-color the non-adjacent pair of nodes inside the AC, again a $(deg+1)$-list coloring problem on a virtual graph by itself. The precise definitions are more subtle than this intuitive overview.  For example, we need to ensure that we can deal with the fact that ACs of different types can occur next to each other in a graph, the special node of a difficult AC could be part of another AC, and the freedom in the two list coloring problems on virtual graphs crucially depends on the order in which we color nodes.  Next, we continue with formally defining the partition.

Let $\psi = \Delta^{1/3}$ and $\phi = \Delta^{2/3}/2$. 
A node $v\notin C$ is \emph{special} for an AC $C$ if it has at least $\phi$ neighbors in $C$. 
A node is \emph{simplicial} if its neighborhood is a clique.
\begin{itemize}
\item An AC is \emph{easy} if it contains a simplicial node or a non-edge (i.e., two non-adjacent vertices).
\item An AC is \emph{difficult} if: a) it is non-easy, b) it has a \emph{special} neighbor, and c) $|C| \ge \Delta-\psi$. For each difficult AC $C$ pick one arbitrary special neighbor.  Non-picked special nodes will not play a special role in the algorithm. 
\item An AC is \emph{nice} if it is easy or contains a picked special node. 
\item An AC is \emph{ordinary} if it is neither difficult nor nice. 
\end{itemize}

Call a special node $v$ an \emph{escape} if it is picked by at least two ACs. If the special node of a difficult AC $C$ is such an escape node denote it by $e_C$, call the clique a \emph{runaway},  and call $e_C$ its \emph{escape}. For a difficult AC whose special node is not picked by another AC, denote its special node by $p_C$, call $C$ a \emph{guarded} AC and call $p_C$ its \emph{protector}. See \Cref{fig:ac-types,fig:guardedRunaway} for examples of these definitions.

\begin{figure}
	\centering
	\includegraphics[width=0.75\textwidth,page=8]{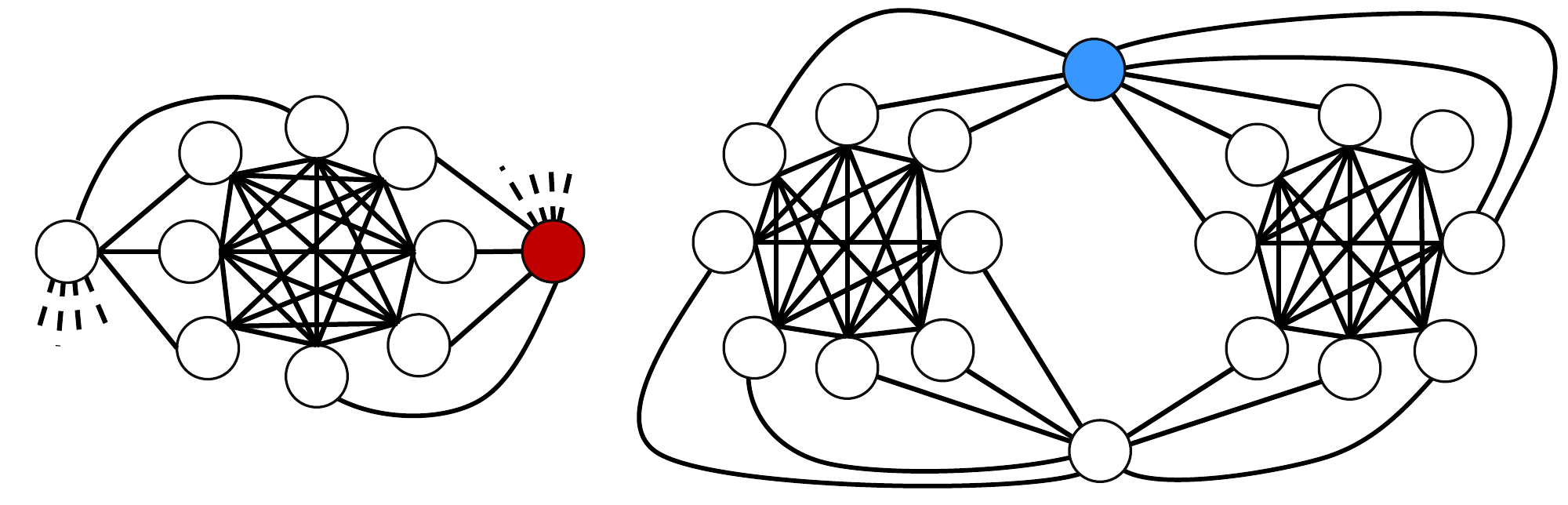}
	\caption{In both of the examples the $\Delta$-cliques have few outside neighbors. On the left side we see a guarded AC with two special nodes. Each of them, e.g., the red node may serve as a protector for the AC. On the right side there are two runaway ACs that have common special nodes. Both special nodes, e.g, the blue node, may become an escape. 
	}
	\label[figure]{fig:guardedRunaway}
\end{figure}

\begin{definition}[Fine Grained ACD]
\label{def:fineGrainedACD}
Given an ACD of a graph, we partition the graph nodes into the following sets.
\begin{compactenum}
	\item $\cP$: the set of protector nodes,
	\item $\cE$: the set of escape nodes, 
	\item $\cV_{*}$: nodes in $V_{sparse}$, excluding those in $\cP\cup \cE$,
	\item $\cO$: nodes in ordinary ACs, 
	\item $\cR$: nodes in runaway ACs,
	\item $\cN$: nodes in nice ACs, excluding those in $\cP\cup \cE$,
	\item $\cG$: nodes in guarded ACs.
\end{compactenum}
\end{definition}

\begin{observation}
\label{obs:difficultNoSpecial}
Difficult and ordinary ACs are cliques and contain neither escapes nor protectors.
\end{observation}
\begin{proof} Ordinary and difficult ACs are cliques as a non-edge between any two of its vertices would classify the respective clique as easy. 
	Let $C$ be a difficult AC and let $v\in C$. As $|C|\geq \Delta-\psi$, and the AC $C$ is not easy,  $v$ has at most $\psi$ external neighbors. As nodes in $\cP\cup \cE$ have at least $\phi>\psi$ neighbors in an AC different from their own, $v\not\in\cP\cup \cE$. 
	
	By definition, any AC that is neither easy nor difficult and contains a special node is classified as a nice AC. Hence,  by definition, (the remaining) ordinary ACs do not contain special nodes, i.e., they contain neither escapes nor protectors. 
\end{proof}
\begin{observation}
	\label{obs:partitionWellDefined}
	The node sets $\cP,\cE, \cV_*, \cO, \cR, \cN,$ and  $\cG$ partition the nodes of the graph.
\end{observation}
\begin{proof}The sets of nodes clearly cover all nodes of the graph as every node in an ACD is either contained in $V_{sparse}$ or in an AC. Most sets are disjoint by definition, in particular, we have $\cP\cap \cE=\emptyset$. The only potential non-empty intersection between sets is covered by \Cref{obs:difficultNoSpecial}, which shows that difficult ACs cannot contain a node of $\cP\cup \cE$.
\end{proof}

\subsection{The Algorithm \& Intuition}
\label{ssec:AlgorithmHighDegree}
During the course of our algorithm we keep updating the palettes of available colors of nodes by removing colors used by already colored neighbors. Such an update can be implemented in a single round. Whenever we create a list coloring instance, the \emph{(uncolored) degree} of a node refers to its degree in the considered subgraph (or in some cases virtual graph), while the \emph{palette} is the set of available colors in the whole graph $G$. 
A node has \emph{slack} $x$ in a subgraph if its palette is at least by $x$ larger than its uncolored degree in the subgraph. A node has \emph{unit-slack} in a subgraph if it has slack of at least one. 
Next, we present our main high-level algorithm, where each step is detailed on later.

\begin{algorithm}[H]\caption{$\Delta$-coloring (High-Level Overview)} \label{alg:orderedPartition}
	\begin{algorithmic}[1] 
		\STATE Compute an ACD ($\eps=1/172$)
		\STATE Break the graph into groups $\cV_*, \cO, \cR, \cN, \cG, \cP$, and $\cE$ according to \Cref{def:fineGrainedACD}. 
		\STATE Run \textsf{SlackGeneration} on $\cV_*\cup \cO\cup \cR$.
		\STATE Color $\cV_*$: sparse nodes in $V_{sparse}\setminus (\cP\cup \cE)$
		\STATE Color $\cO$: nodes in ordinary ACs
		\STATE Color $\cR$: nodes in runaway ACs 
		\STATE Color $\cN$: nodes in nice ACs, excluding those in $\cP\cup \cE$ 
		\STATE Color $\cG \cup \cP$: nodes in guarded ACs $\cG$ and protector nodes $\cP$
		\STATE Color $\cE$: escape nodes
	\end{algorithmic}
	Note that some nodes of $\cV_*\cup \cO\cup \cR$ are already colored in \textsf{SlackGeneration}.
\end{algorithm}
See \Cref{figure:ACorder} on page \pageref{figure:ACorder} for an illustration for the dependencies on the coloring order of different node sets. 

\subsection{Our Steps in a Nutshell}

In order to use the $O(\logstar n)$-round algorithm of \Cref{lem:listColoring}, we need not only ensure that the nodes in the list coloring instances have unit-slack, but that they also have large palettes.  One observation that we use frequently is that a node with many uncolored neighbors always has  a large palette of remaining colors.
Thus, we achieve large palettes in one of three ways: many same-colored pairs of neighbors, many stalled neighbors, or by coloring large-degree subsets simultaneously. We next provide intuition for why each of the steps works.

Step~1 uses the algorithm from \Cref{lem:acd} and the partition in Step~2 is well defined due to \Cref{obs:partitionWellDefined}. We continue with the remaining steps.
\myparagraph{Step~3 (Slack Generation):} Step~3 executes $\slackgeneration$ (recall its description in \Cref{ssec:challenges} or see \Cref{sec:slackGeneration} for the formal definition) where only nodes in $\cV_*\cup \cO\cup\cR$ participate. We prove the following lemma that is central to our algorithm. It abstracts away all randomized steps and probabilistic guarantees that we need for the rest of the reduction. Its proof appears in \Cref{sec:slackGeneration}.
\begin{restatable}[Properties of Step~3, \textsf{SlackGeneration}]{lemma}{lemSlackGenerationProperties}
	\label{lem:slackGenerationProperties}
	If $\Delta=\omega(\log^{3}n)$, the following properties hold w.h.p. after Step~3.
	\begin{compactenum}[(i)]
		\item All sparse nodes have slack $\Omega(\Delta)$ in $G[\cV_* \cup \cO]$,  and escape nodes have slack $\Omega(\Delta)$ in $G[V]$.
		\item  All ordinary cliques have  $\Omega(\psi)$ uncolored vertices with unit-slack in $G[\cV_* \cup \cO]$.
		\item  
		For each difficult AC $C$, at most half the vertices in $N_{C}(e_C)$ ($N_{C}(p_C)$) are colored.
	\end{compactenum} 
	\label{lem:slackprop}
\end{restatable}
We next provide some intuition for the result of \Cref{lem:slackGenerationProperties}. The result for sparse nodes is closely related to similar results in previous work \cite{Molloy1998,EPS15,HKMT21}. By definition, a sparse node has  $\Omega(\Delta^2)$ non-edges in its neighborhood, and for any of these the probability that both "endpoints" are same-colored in $\slackgeneration$ is $\Theta(1/\Delta)$. Thus, in expectation, $\Omega(\Delta)$ of these non-edges are colored monochromatic, providing $\Omega(\Delta)$ permanent slack for the sparse node. The dependencies can be taken care of with a version of Talagrand's inequality. 
The most involved part of \Cref{lem:slackGenerationProperties} is part $(ii)$, that provides slack for ordinary cliques. First, we show that ordinary ACs  have $\Omega(\phi/\psi)=\Omega(\psi)$ distinct outside neighbors. The choice of $\phi$ and $\psi$ is crucial in this step. 
This merely follows  as they do not have a special neighbor that is incident to many nodes of the AC. Then, consider one of these outside neighbors $w$. Due to \Cref{lem:acd}, 4., it has at most $(1-2\eps)\Delta$ neighbors inside the clique, implying that it has a set $A_w$ of $2\eps\Delta$ non-neighbors in the AC. If there were no dependencies between the coloring of different nodes, it is easy to see that there is a constant probability that $w$ and a node of $A_w$ get same-colored. Then, any uncolored neighbor $c$ of $w$ inside the clique (and these exist) obtains permanent slack from this same-colored pair. Intuitively, a similar reasoning holds for all the other $\Omega(\psi)$ outside neighbors, yielding $\Omega(\psi)$ nodes with permanent slack for the clique. The main difficulty here in the randomized analysis is to deal with the various types of dependencies between the different events. 

Property $(iii)$ follows with a Chernoff bound as every node only participates in the coloring procedure with a small constant probability. 

\myparagraph{Step~4 (Sparse Nodes):} As \Cref{lem:slackGenerationProperties} (i) provides $\Omega(\Delta)$ permanent slack for  sparse nodes, coloring these nodes immediately yields a $(deg+1)$-list coloring instance with lists of size $\Omega(\Delta)$. 

\myparagraph{Step~5 (Ordinary ACs):} \Cref{lem:slackGenerationProperties} (ii) provides us with a large set $X$ of uncolored nodes inside the ordinary AC that all have permanent unit-slack. We stall the nodes in $X$ and first color the remaining nodes of the AC which have temporary slack due to having many neighbors in $X$. Then, we proceed with coloring the nodes in $X$. In both cases the $(deg+1)$-list coloring instances have large lists because every node is connected to all nodes in $X$ and $X$ is large.  

\begin{figure}
	\centering
	\includegraphics[width=0.6\textwidth]{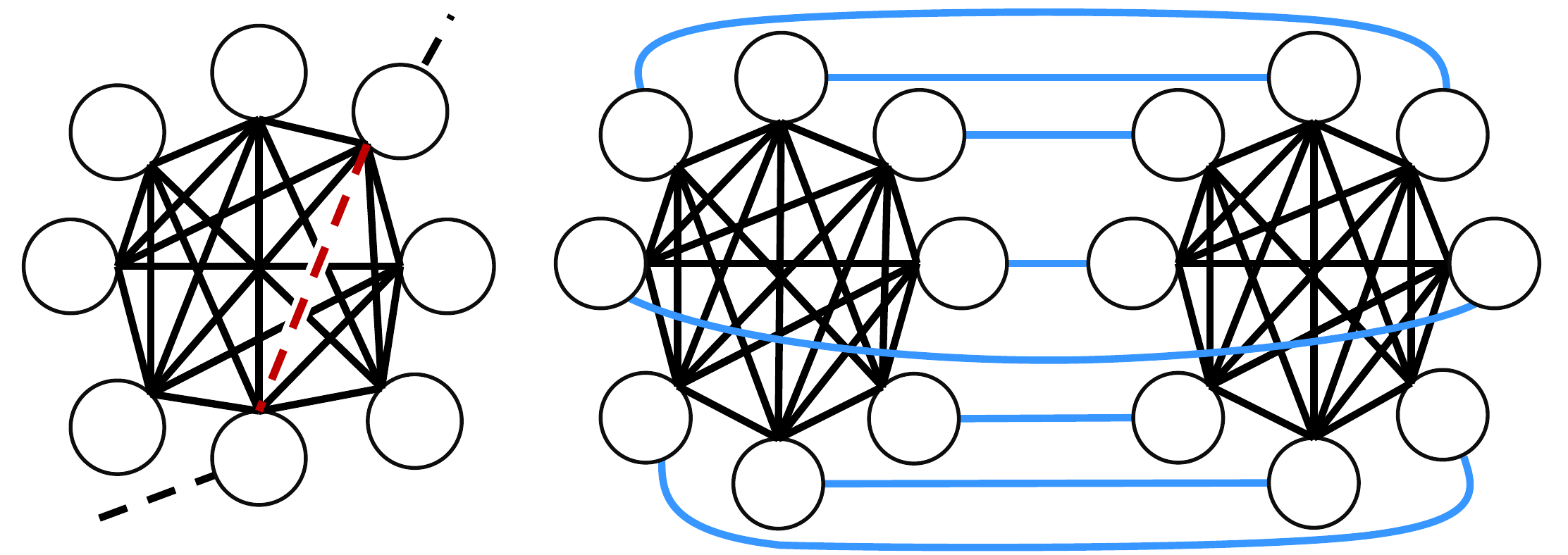}
	\caption{Two of the challenging situations that can occur. On the left, a $\Delta+1$-clique, minus a single edge (red dotted). On the right,
		two $\Delta$-cliques connected via a perfect matching (blue). The left hand side is a nice AC, the right hand side is an ordinary AC.}
	\label[figure]{fig:ac-types}
\end{figure}

\myparagraph{Step~6 (Runaway ACs):} For each runaway AC $C$ there is an escape node $e_C$ that is adjacent to at least $\phi$ nodes in $C$. 
 By \Cref{lem:slackGenerationProperties} (iii), at least $\phi/2$ of them remain uncolored during \textsf{SlackGeneration},
As $e_C\in \cE$ is only colored in Step~9, these nodes obtain temporary slack. Then, all nodes of the clique can be colored with two list coloring instances with a similar two-layer approach as for ordinary ACs.  See \Cref{fig:guardedRunaway} (right hand side),  for an illustration of these ACs.

\myparagraph{Step~7 (Nice ACs):} For nice ACs there are several subcases that we deal with one after the other. Nice ACs that contain a protector or an escape node can be colored in two  $(deg+1)$-list coloring instances, as most nodes have temporary slack from the stalled vertices (protectors and escapes are only colored in later steps). If a nice AC has no non-edge and no node in $\cP\cup \cE$,  we can show that the nice clique has to contain a node with degree $<\Delta$, that can serve as a toehold. 
The most involved part, is the case of having a non-edge and no node of $\cP\cup\cE$. In this case we form a \emph{pair} from the endpoints
%of the two "endpoints" 
of the non-edge and create a virtual graph $H$ in which each pair forms a vertex, and two pairs are adjacent if some of their endpoints are adjacent in $G$. Additionally, each pair receives as its list of available colors, the intersection of the lists of its endpoints. Then, list-coloring $H$ same-colors the endpoints in each pair. This list coloring problem is solvable because (1) one can show that both endpoints have more than $\Delta/2$ common uncolored neighbors in their clique, yielding a list of size $\Delta/2$, and because (2)  one can show that the maximum degree of $H$ is at most $O(\eps\Delta)\ll \Delta/2$. The latter follows as \Cref{obs:outsideAndAntiDegree} bounds the outside degree of nodes inside the AC  by $8\eps \Delta$, and hence the degrees in $H$ by $16\eps\Delta$. Once we have same-colored the pairs in these nice ACs there are several nodes with permanent slack in each nice AC allowing to color the ACs with a two layer approach, similar to the previous cases. Note that a bit more care is needed as there may be further non-edges yielding that nice ACs are less connected than ordinary ACs or runaway ACs.

\myparagraph{Step~8 (Guarded ACs):}  See \Cref{fig:guardedRunaway} (left hand side),  for an illustration of these ACs.  Similar to the treatment of nice ACs, we form pairs of non-adjacent vertices. Then, same-coloring these pairs yields permanent slack for the common neighbors of vertices in the pair. Instead of pairing two non-adjacent vertices inside the AC $C$ (these do not exist!), we form pairs by teaming up the protector node $p_C\notin C$ with a non-neighbor $u_C$ inside the clique. Such a non-neighbor exists, because a protector node can have at most $(1-2\eps)\Delta<|C|$ neighbors in $C$, due to \Cref{lem:acd}, 4. Similar to the treatment of nice ACs, same-coloring the pairs yields a list coloring instance on a virtual graph $H$. 	The main challenge here is to show that these pairs of nodes can be colored consistently (and in parallel with the pairs of other guarded ACs) with a color in their joint palette. Here, the definition of protector nodes and the choice  of $\psi$ and $\phi$ is crucial in order to show that this list coloring instance can be solved efficiently. Here, it is also crucial that a protector node is only used by a single runaway clique. This is why special nodes that several difficult cliques want to use are treated differently (as escape nodes). As this step is crucial, we provide the full proof.

\begin{restatable*}{lemma}{lemGuardedPairs}
	\label{lem:guardedPair-coloring}
	The set $\{(u_C,p_C)\}_C$ of pairs for guarded ACs induces a $(deg+1)$-list coloring instance with palette size $\geq \phi/2$. 
\end{restatable*}
\begin{proof}
A node of $G$ \emph{conflicts} with a pair $(u_C,p_C)$ of a guarded AC $C$ if the node is already colored or is contained in an adjacent pair of the same phase.
	Consider a guarded AC $C$. As $C$ has no non-edges and its size is at least $\Delta-\psi$ node $u_C$ has at most $\psi$ neighbors outside $C$.  By definition, node $p_C$ has at most $\Delta - \phi$ neighbors outside $C$. Due to \Cref{obs:difficultNoSpecial}, $C$ does not contain any protector of another AC. Hence, no nodes of $C$ (besides $u_C$) are colored in this phase. Thus, at most $\Delta - \phi + \psi \le \Delta - \phi/2$ nodes conflict with the pair $(u_C,p_C)$. The number of nodes that conflict with a pair is significantly less than $\Delta$, the number of colors initially available. Thus, the problem of coloring the pairs reduces to $(deg+1)$-list coloring with a palette of size at least $\phi/2$.
\end{proof}

Once we have colored the pair $(u_C,p_C)$, all nodes in the set $X$ of neighbors of $p_C$ in $C$ obtain temporary unit slack. Indeed $X$ contains at least $\phi/2$ nodes as $p_C$ has at least $\phi$ neighbors  by \Cref{lem:slackGenerationProperties} $(iii)$ half of these are still uncolored.  Now, we put the nodes in $X$ aside, providing temporary slack for all remaining nodes in $C$, list color these, and then list color the nodes in $X$. One can show that each list coloring instance uses lists of size at least $\phi/2$. 

In this step we also automatically color all protector nodes. 

\myparagraph{Step~9 (Escape Nodes):} While it is crucial for the previous steps that escape nodes are colored at the very end, coloring them is just a single $(deg+1)$-list coloring instance. Due to \Cref{lem:slackGenerationProperties}, $(i)$, each escape node has $\Omega(\Delta)$ permanent slack, which implies that the list coloring problem is solvable and that all lists are large.

\section{\texorpdfstring{Ultrafast $\Delta$-coloring: Details of the Reduction}{Ultrafast Delta-coloring: Details of the Reduction}}
\label{sec:detailedReduction}

To formally prove the correctness and runtime of our reduction we use  the following nomenclature. 

\begin{definition} Consider an arbitrary step of the algorithm. A node is \emph{white} if it has unit-slack or a neighbor that will be colored in a later step of the algorithm. A node is \emph{gray} if it has a white neighbor (and is not white itself). 
\end{definition}
We emphasize that the property of being white or gray crucially depends on the order in which nodes are processed. Whenever we claim that a node (set) is white or gray it refers to a certain step of the algorithm in which we color the respective set. 

Any set consisting of white and gray nodes only can be colored. 
\begin{lemma}
	\label{lem:graytoneColorable}
	A set consisting of white and gray nodes only can be colored via two $(deg+1)$-list coloring instances. 
\end{lemma}
\begin{proof}
	First, we use one $(deg+1)$-list coloring instances for the gray nodes,  and then one for the white nodes.  Gray nodes have unit-slack in the graph induced by themselves as they have at least one white non-participating neighbor. White nodes have unit-slack or a non-participating neighbor  in the graph induced by all white nodes.
\end{proof}
We say that a set of white or gray nodes has list size at least $\ell$ if all lists in the respective $(deg+1)$-list coloring instances from \Cref{lem:graytoneColorable}  have this minimum size.

\subsection{Detailed Steps \& Analysis}
\label{ssec:detailedSteps}
We now prove the correctness and the constant runtime bounds of each step in detail. As the list sizes crucially depend on the choice of $\phi$ and $\psi$ and are essential to obtain our $O(\logstar n)$, we prominently  mention the list size of each step, including all occurring $(deg+1)$-list coloring instances.

\subsubsection{Step~1: ACD Construction}
Step~1 uses the algorithm from \Cref{lem:acd} that runs in $O(1)$ rounds and requires that $\Delta=\omega(\log^2 n)$.  

\subsubsection{Step~2: Partition of the Nodes}

The partition in Step~2 is well defined partition due to \Cref{obs:partitionWellDefined}. Nodes can determine in which set they are in $O(1)$ rounds, after the ACD has been computed as each AC receives a unique identifier, e.g., the maximum ID of its nodes. 

\subsubsection{Step~3: \textsf{SlackGeneration}}
Step~3 consists of a single round of color trials, in which each node of $\cV_*\cup\cO\cup\cR$ activates itself with a small constant probability and activated nodes then try a color from $[\Delta]$ uniformly at random; a color is kept permanently by a node if no neighbor tries the same color, otherwise the trial is discarded.We show the following lemma. 
\lemSlackGenerationProperties*

\Cref{lem:slackGenerationProperties} abstracts away all randomized steps and probabilistic guarantees that we need for the rest of the reduction. Hence, to ease the presentation we defer its involved proof to \Cref{sec:slackGeneration}.

\subsubsection{Step~4: Coloring Sparse Nodes}
Sparse nodes are white due to \Cref{lem:slackGenerationProperties}.
\begin{lemma}
	\label{lem:basic-graytone}
	W.h.p, sparse nodes  are white with list size at least $\Omega(\Delta)$.
\end{lemma}
\begin{proof}
	By \Cref{lem:slackGenerationProperties} (i) a sparse node  has slack $\Omega(\Delta)$ in $G[\cV_*\cup \cO]$. As all nodes not in $\cV_*\cup \cO$ are colored later, each sparse node is white and lists are of size $\Omega(\Delta)$.
\end{proof}
	
\paragraph{List sizes:}
\begin{compactitem}
    \item Sparse nodes: $\Omega(\Delta)$
\end{compactitem}

\subsubsection{Step 5: Ordinary Cliques}
Nodes in ordinary cliques are white or gray due to \Cref{lem:slackGenerationProperties}.
\begin{lemma}
	W.h.p, all nodes in
	ordinary ACs are white or gray with list size at least $\Omega(\psi)$.
 \label{lem:ordinary-white}
\end{lemma}
	\begin{proof}
	By \Cref{lem:slackGenerationProperties} (ii), each ordinary AC has $\Omega(\psi)$ uncolored nodes with unit-slack in $G[\cV_*\cup \cO]$. These nodes are white when we process ordinary ACs. The remaining uncolored nodes of the AC are then gray, since $C$ is a clique.
\end{proof}

The following list sizes follow as each node of an AC is adjacent to each of the uncolored unit-slack nodes inside the clique.  
\paragraph{List sizes:}
\begin{compactitem}
\item Ordinary ACs (gray instance): $\Omega(\psi)$,
   \item Ordinary ACs (white instance): $\Omega(\psi)$.
\end{compactitem}

\subsubsection{Step 6: Runaway Cliques}
	\begin{lemma}
	W.h.p, all nodes in runaway ACs are white or gray with list size at least $\phi/2$.
 \label{lem:runawaylists}
	\end{lemma}
	\begin{proof}
	A runaway AC $C$ with escape $e_C$ contains at least $\phi$ neighbors of $e_C$ by definition. By \Cref{lem:slackGenerationProperties} (iii), at least $\phi/2$ of them remain uncolored during \textsf{SlackGeneration}, these nodes are white, because $e_C$ is stalled and only colored in the last step (Step~9). The remaining uncolored nodes of $C$ are then gray, since $C$ is a clique.
\end{proof}

\paragraph{List sizes:}
\begin{compactitem}
    \item Runaway AC (gray instance): $\phi/2$, 
    \item Runaway AC (white instance): $\phi/2$.
\end{compactitem}

\subsubsection{Step~7: Coloring Nice ACs} 
We color the nice ACs in the following order. 
\begin{compactenum}
	\item Nice ACs containing an escape, or a protector,
	\item Nice ACs with no non-edge and no escape nor protector node,
	\item Nice ACs with a non-edge.
\end{compactenum}
The first two steps reduce to few $(deg+1)$-list coloring instances via the next lemma.
\begin{lemma} 
	\label{lem:niceAC1}
	All nodes of a nice AC that contains an escape/protector or has no  escape/protector and no non-edge are white or gray with list size at least $\Delta/2$. 
\end{lemma}
\begin{proof}
	Protectors and escape nodes are stalled and colored after the nodes in nice ACs. Let $v$ be a protector or an escape node of $C$ and let $W\subseteq C$ be all neighbors of $v$ in $C$. Note that $W$ may contain protectors and escape nodes as well. Due to \Cref{lem:acd}, (3), $W$ has size at least $(1-4\epsilon)\Delta$. By \Cref{obs:outsideAndAntiDegree}, each node of $C$ has at most $7\eps\Delta$ anti-neighbors in $C$. Hence, any node in $C$ has at least $(1-11\eps)\Delta \ge \Delta/2$ neighbors in $W$. Hence, all nodes in $C\setminus (W\cup \cP\cup\cE)$ are gray with list size at least $\Delta/2$ and all nodes in $W\setminus \cP\cup\cE$ are white with list size at least $\Delta/2$.
	
	If a nice AC $C$ contains no non-edge, no protector, and no escape, it contains a simplicial node $u$, by the definition of nice ACs. Node $u$ must have zero outside degree. Furthermore, $C$ is a clique, and since by assumption the graph contains no $\Delta+1$-cliques, $u$ has at most $\Delta-1$ neighbors, all inside $C$. Then $u$ white, but since it is isolated in the graph it can be colored independently of the rest of the graph after the gray set $C\setminus \{u\}$ has been colored.
\end{proof}

For the remaining nice ACs we cannot immediately prove that they  only contain white and gray nodes. Instead, we manually create a white toehold for them by coloring a carefully selected pair of nodes with the same color. We then argue that all nodes in the AC are white or gray, which allows us to complete the coloring.

\textit{Toeholds for nice ACs:} For a nice AC $C$ with a non-edge and no escape/protector, let $(u_C, w_C)$ be any non-edge in $C$.

\begin{lemma}
	The set $\{(u_C,w_C)\}_C$ of pairs induces a $(deg+1)$-list coloring  instance with list size at least $\Delta/2$.
	\label{lem:niceAC3}
\end{lemma}
\begin{proof}
	We say that a node of $G$ \emph{conflicts} with a pair if the node is already colored or is contained in an adjacent pair of the same phase.
	Consider a pair $(u_C,w_C)$ where $C$ is nice. $C$ is not colored during slack generation and the nodes in the pair are the only nodes in $C$ colored while coloring the pairs of such nice ACs as $C$ does not contain a protector. Thus, the only nodes that conflict with the pair are the outside neighbors of either node, of which there are at most $2\cdot 4\epsilon\Delta \le \Delta/2$, by \Cref{obs:outsideAndAntiDegree}. Hence, the number of nodes that conflict with a pair is significantly less than $\Delta$, the number of colors initially available. Coloring the pairs reduces to $(deg+1)$-list coloring with palettes of size $\geq \Delta/2$.
\end{proof}

\begin{lemma}
	\label{lem:niceAC4}	
	The pairs of nice ACs can be colored in the \CONGEST model in $O(\log^3 \log n)$ rounds, and even in $O(\log^* n)$ rounds when $\Delta = \omega(\log^7 n)$.
\end{lemma}
\begin{proof}
	Both nodes of the pair $(u_C, w_C)$ have
	at least $(1-4\epsilon)\Delta$ neighbors in $C$, so they have at least $(1-4\epsilon)\Delta - (|C| - (1-4\epsilon)\Delta) > (1-11\epsilon)\Delta \ge \Delta/2$ common neighbors in $C$.
	They provide the bandwidth to transmit to one node all the colors adjacent to the other node.  Also, all messages to and from $u_C$ \emph{vis-a-vis} its external neighbors can be forwarded in two rounds. Hence, we can simulate any \CONGEST coloring algorithm on the pairs with $O(1)$-factor slowdown, in particular, we can simulate the algorithm from \Cref{lem:listColoring}. Additionally, \Cref{lem:niceAC3} shows that the palettes of the list coloring instance have at least $\Delta/2=\omega(\log^7 n)$ colors if $\Delta=\omega(\log^{7}n)$. In this case, the algorithm runs in $O(\logstar n)$ rounds.
\end{proof}

After coloring the pairs, we can color the remaining nice ACs. 

\begin{lemma}
	\label{lem:niceAC5}
	After coloring the pairs, the nodes in respective nice  ACs are white or gray with list sizes at least $\Delta/3$.
\end{lemma}
\begin{proof}
	Consider such a nice AC $C$. 
	Recall from the proof of \cref{lem:niceAC4} that there are at least $\Delta/2$ common neighbors of the nodes in the pair $(u_C,w_C)$ of a nice AC $C$. Denote these nodes by $W$ and observe that they are white.  As the anti-degree is bounded by $7\eps\Delta$ (see \Cref{obs:outsideAndAntiDegree}), each node of $C$  has at least $\Delta/2 - 7\epsilon\Delta = \Delta/3$ neighbors in $W$. Thus all nodes in $C\setminus W$ are gray with the claimed list size. The latter argument also shows that the list size for nodes in $W$ is at least $\Delta/3$. 
\end{proof}

\paragraph{List sizes:}
\begin{compactitem}
	\item Nice ACs containing an escape, or a protector (gray instance): $\Delta/2$ 
	\item Nice ACs containing an escape, or a protector (white instance): $\Delta/2$  
	\item Nice ACs with no non-edge and no escape nor protector node (gray instance): $|C|\geq \Delta/2$.  
	\item Nice ACs with no non-edge and no escape nor protector node (white instance): Nodes are simplicial, have zero outside degree and are colored independently from the rest of the graph. 
	\item Nice ACs with a non-edge (pairs): $\Delta/2$. 
	\item Nice ACs with a non-edge (gray): $\Delta/3$.
	\item Nice ACs with a non-edge (white): $\Delta/3$.
\end{compactitem}

\subsubsection{Step~8: Coloring Guarded ACs Including their Protectors}
Similarly to the last processed class of nice ACs, we cannot immediately argue that all nodes in guarded ACs are white or gray.  Instead, we first manually create a white toehold for them by coloring a carefully selected pair of nodes with the same color. We then argue that all remaining nodes in the AC are gray, which allows us to complete the coloring.

\textit{Toehold for guarded ACs:}  For a guarded AC $C$, let $p_C$ be its protector and $u_C$ be an arbitrary uncolored non-neighbor of $p_C$ in $C$. Node $u_C$ exists as all nodes of $C$ are yet to be colored, $p_C\notin C$ has at most $(1-2\eps\Delta)$ neighbors in $C$ by \Cref{lem:acd} (4), and $|C|\geq (1-\eps)\Delta$. 

\lemGuardedPairs
\begin{proof}[Repeated proof]
A node of $G$ \emph{conflicts} with a pair $(u_C,p_C)$ of a guarded AC $C$ if the node is already colored or is contained in an adjacent pair of the same phase.
	Consider a guarded AC $C$. As $C$ has no non-edges and its size is at least $\Delta-\psi$ node $u_C$ has at most $\psi$ neighbors outside $C$.  By definition, node $p_C$ has at most $\Delta - \phi$ neighbors outside $C$. Due to \Cref{obs:difficultNoSpecial}, $C$ does not contain any protector of another AC. Hence, no nodes of $C$ (besides $u_C$) are colored in this phase as $C$. Thus, at most $\Delta - \phi + \psi \le \Delta - \phi/2$ nodes conflict with the pair $(u_C,p_C)$. The number of nodes that conflict with a pair is significantly less than $\Delta$, the number of colors initially available. Thus, the problem of coloring the pairs reduces to $(deg+1)$-list coloring with a palette of size at least $\phi/2$.
\end{proof}

To color the pairs via \Cref{lem:listColoring} in $O(\logstar n)$ rounds, we only require $\Delta=\omega(\log^{14} n)$, as the list size of the pairs is at least $\phi/2=\omega(\log^7n)$ for that choice of $\Delta$. 
\begin{lemma}
	\label{lem:guardedPairs-col-congest}
	The pairs of guarded ACs can be colored in the \CONGEST model in $O(\log^3 \log n)$ rounds, and even in $O(\log^* n)$ rounds when $\Delta = \omega(\log^{14} n)$.
\end{lemma}

\begin{proof}
	For a guarded AC $C$, the pair $\{u_C,p_C\}$ has at least $\phi$ common neighbors in $C$, by definition (recall that $C$ is a clique).
	This suffices to communicate to $p_C$ the colors of the at most $\psi$ external neighbors of $u_C$. The palette of valid colors for the pair can then be maintained at $p_C$. Also, all messages to and from $u_C$ \emph{vis-a-vis} its external neighbors can be forwarded in two rounds. Hence, we can simulate any \CONGEST coloring algorithm on the pairs with an $O(1)$ slowdown, in particular, the algorithm from \Cref{lem:listColoring}.
	Also, \Cref{lem:guardedPair-coloring} shows that the palettes of the list coloring instance has at least $\phi/2=\Omega(\log^7 n)$ colors the algorithm runs in $O(\logstar n)$ rounds,  if $\Delta=\Omega(\log^{14}n)$. 
\end{proof}

We proceed with coloring guarded ACs.
The protectors are colored as part of coloring the pairs.

\begin{lemma}
	After coloring the pairs, all nodes in guarded ACs are white  or gray with list sizes at least $\phi$ and $\Delta/2$, respectively. 	\label{lem:nice-guarded-graytone}
\end{lemma}
\begin{proof}
	Consider a guarded AC $C$. It can neither contain a protector nor an escape (\Cref{obs:difficultNoSpecial}).  Guarded ACs are uncolored, except for the node $u_C$. By definition the protector $p_C$ has $\phi$ neighbors in $C$. As $C$ forms a clique, each of these nodes is a common neighbor of $p_C$ and $u_C$ and thus white, denoted these by $W$. As $C$ is a clique all nodes of $C$ have at least $\phi$ neighbors in $W$ and all nodes of $C\setminus W$ are gray. When coloring the gray set all nodes of the AC, except for $u$,  are uncolored and since there are no non-edges ($|C|\leq \Delta$) the list is is $|C|-1\geq (1-\eps)\Delta-1\geq \Delta/2$. 
\end{proof}

\paragraph{List sizes:}
\begin{compactitem}
    \item Guarded ACs (pairs): $\phi/2$,
    \item Guarded ACs (gray instance): $\Delta/2$,
    \item Guarded ACs (white instance): $\phi$.
\end{compactitem}

\subsubsection{Step 9: Coloring Escape Nodes}

\begin{lemma}
	\label{lem:escape-white}
	Escape nodes are white with list size at least $\Omega(\Delta)$.
\end{lemma}
\begin{proof}
	 \Cref{lem:slackGenerationProperties}, an escape node has slack $\Omega(\Delta)$ in the whole graph $G[V]$, thus it is white when it becomes colored. 
\end{proof}

\paragraph{List size(s):}
\begin{compactitem}
    \item Escape nodes: $\Omega(\Delta)$.
\end{compactitem}

\subsection{Proof of  Theorem~\ref{thm:mainLogstar} and Corollary~\ref{corr:mainPolylog}}
\label{app:MainProof}
\begin{restatable}[reduction to list coloring]{theorem}{thmReduction} 
\label{thm:ListColorReduction}
There is a constant time randomized \CONGEST algorithm that w.h.p.\ reduces the problem of $\Delta$-coloring an $n$-node graph with maximum degree $\Delta\geq \omega(\log^3 n)$ to a constant number of $(deg+1)$-list coloring instances.
\end{restatable}
\begin{proof}[Proof of \Cref{thm:ListColorReduction}]
Computing the ACD runs in $O(1)$ rounds via \Cref{lem:acd}.  We have already reasoned that Step~2 can be executed in $O(1)$ rounds.  Executing $\slackgeneration$ in Step~3 works in $O(1)$ rounds as it only consists of a single round of color trials, see \Cref{sec:slackGeneration}.
If $\Delta = \omega(\log^{3} n)$, the probabilistic claims of \Cref{lem:slackGenerationProperties} hold w.h.p. Setting up the constant different $(deg+1)$-list coloring instances can be done in constant time. 
\end{proof}

\begin{remark}
\label{rem:sequentialInstances}
The $(deg+1)$-list coloring instances of \Cref{thm:ListColorReduction} need to be executed in the specified order as the lists of later instances depend on the solution computed for earlier processed instances. 

We emphasize that it may be possible to reduce to fewer sequentially executed list coloring instances as not every instance depends on the output of all instances that we process earlier. However, we did not  optimize the anyway constant number of instances to keep the arguments for the list size and solvability of each instance as simple as possible. 
\end{remark}
\tmMainLogstar*
\begin{proof}[Proof of \Cref{thm:mainLogstar}]
The constant number of sequentially (see \Cref{rem:sequentialInstances}) executed $(deg+1)$-list coloring instances can each be solved in $\poly\log\log n$ \CONGEST rounds  via \Cref{lem:listColoring}. 
	
As explicitly stated in the previous section, the list sizes of the constant different $(deg+1)$-list coloring instances are of size at least $\Omega(\Delta)$ (sparse nodes, escape nodes, all instances for nice cliques, guarded ACs gray instance), $\phi/2$ (pairs of guarded ACs, white instance of guarded ACs, runaway ACs),  $\Omega(\phi/\psi)=\Omega(\psi)$ (ordinary ACs), where the smallest list size is attained by the last term. Hence, if $\Delta=\omega(\log^{21}n)$, we obtain that $\Omega(\psi)=\omega(\log^7 n)$. In this case, each usage of \Cref{lem:listColoring} runs in $O(\logstar n)$ rounds.

Actually, note that coloring the respective pairs in Step~7 (nice ACs) and Step~8 (guarded ACs) happens on virtual graphs and it is not immediate that \Cref{lem:listColoring} can be used in a black box manner on these virtual graphs. Hence, we present a white box treatment of these cases in \Cref{lem:niceAC4,lem:guardedPairs-col-congest}.
\end{proof}

\begin{proof}[Proof of \Cref{corr:mainPolylog}]
If $\Delta=\omega(\log^3 n)$, we use \Cref{thm:mainLogstar}.
	If $\Delta= O(\log^3 n)$ we use the \LOCAL algorithm from \cite{GHKM21,RG20}. For this bound on $\Delta$, this algorithm has runtime $O(\log \Delta)+\poly\log\log n=\poly\log\log n$. As stated in \cite{GHKM21}, the algorithm works if the graph has no $\Delta+1$-clique and $\Delta\geq 4$. For $\Delta=3$, the work presents a different algorithm that solves the problem $O(\log^2\log n)$ rounds.
\end{proof}

\section{Slack Generation}
\label{sec:slackGeneration}

We now prove \cref{lem:slackprop} on the properties of Step~3 of \Cref{alg:orderedPartition}. 
We first state its simple algorithm \textsf{SlackGeneration}.

\begin{algorithm}[H]\caption{{\slackgeneration} (vertex $v$)}\label{alg:slackgen}
\begin{algorithmic}[1]
	\STATE Become \textsf{activated} with probability $p_g=1/20$.
	\STATE{\textbf{if}} \textsf{activated} \textbf{then}  $\trycolor(v,c_v)$ where $c_v$ is chosen u.a.r.\ from $[\Delta]$
\end{algorithmic}
\end{algorithm}
In $\trycolor(v,c_v)$ node $v$ sends color $c_v$ to all of its neighbors and receives the set of colors $T=\{c_u : u\in N(v)\}$ tried by its neighbors. If $c_v\notin T$ then $v$ gets permanently colored with $c_v$; otherwise $v$ remains uncolored. 

A node of degree $\Delta$ with local sparsity $\zeta_v$ has $\Delta\zeta_v$ non-edges in its neighborhood. A fixed pair of non-adjacent nodes $u,w\in N(v)$ selects the same color in \textsf{SlackGeneration} with probability $\Theta(1/\Delta)$. Including doublecounting, $v$ would thus expect $\Omega(\zeta_v)$ pairs of nodes colored with the same color.
It has been shown that this intuition is correct: $v$ obtains slack $\Omega(\zeta_v)$ w.h.p., if $\zeta_v$ is large enough.  
We need here the following slightly more general statement for when not all nodes participate in \textsf{SlackGeneration}. 

\begin{lemma}[Slack Generation, \cite{Molloy1998,EPS15,HKMT21}]
\label{lem:slackGeneration}
Let $S\subseteq V$. If a node $v$ has sparsity $\zeta_v$ in $G[V]$, then after  $\textsf{SlackGeneration}(S)$, the palette of $v$ exceeds its uncolored degree in $G[S]$ by $\Omega(\zeta_v)$, with probability $1-\exp(-\Omega(\zeta_v))$. 
\end{lemma}
\begin{proof}
Let $x=\Delta-deg_S(v)$ be the number of colors by which the initial palette exceeds the initial degree towards $S$. 
If $x> \zeta_v/2$, then we get deterministic slack of $x$. If $x\leq \zeta_v/2$, then $v$ has sparsity $\zeta_v-x$ in $G[S]$ as each node $u$ not in $S$ can reduce the sparsity by at most by $\bar{m}(u)/\Delta\leq 1$, where $\bar{m}(u)$ is the number of non-edges in $G[N(v)]$ incident to $u$. Thus, 
we get slack $\Omega(\zeta_v)$ w.p.\ at least $1-\exp(-\Omega(\zeta_v))$ by \cite[Lemma 3.1]{EPS15} applied to the graph $G[S]$. 
\end{proof}

Combining the above lemma with properties of the almost-clique decomposition yields the following.
\begin{corollary}
\label{cor:externalDegreeSlack}
Let $C$ be a non-easy clique, $v \in C$ and $S \supseteq C(v)$.
Then, after $\textsf{SlackGeneration}(S)$, $v$ has slack $\Omega(e(v))$ in $G[S]$ with probability $1-\exp(-\Omega(e(v)))$.
\end{corollary}
\begin{proof}
Each $u \not\in C(v)$ has fewer than $(1-2\eps)\Delta$ neighbors in $C(v)$, whereas $v$ has at least $|C(v)|-1 \ge (1-\eps)\Delta$ neighbors in $C(v)$, by \cref{lem:acd} (1,4) and the fact that $C$ is a clique.
So, $u$ is a non-neighbor of at least $\eps \Delta$ neighbors of $v$. 
Thus,
$\zeta_v \ge \frac{e(v) |N(v) \setminus N(u)|}{\Delta} \ge  e(v) \eps$.
The corollary then follows from \cref{lem:slackGeneration}.
\end{proof}

We continue (\Cref{obs:vertexCoverMatching}, \Cref{lem:ordinaryMatchingSlack,lem:ordinarySlack}) with showing that ordinary ACs obtain sufficiently many nodes with unit-slack. We begin with a structural observation  about the nodes in such ACs in the case that the AC is sufficiently large. The observation holds deterministically and we only need it for the sake of analysis.

\begin{observation}
\label{obs:vertexCoverMatching}
For any ordinary clique of size $|C|\geq \Delta-\psi$ there exists a matching of size $\Omega(\psi)$ in the bipartite graph between $Y=N(C)\setminus C$ and $C$.
\end{observation}
\begin{proof}
Since $C$ is not easy, it contains no simplicial node.
Thus, each node of $C$ has degree at least one in this bipartite graph.
Also, since $C$ is not difficult (but satisfies both $|C|\geq \Delta-\psi$ and being non-easy), it has no special node.
Thus, each node of $Y$ has at most $\phi$ neighbors in this bipartite graph.
Therefore, a greedy matching in the bipartite graph has at least 
$\min\{|Y|, |C'|/\psi\} = \Omega(\Delta/\phi) = \Omega(\psi)$ edges.
\end{proof}

Talagrand's inequality has become a standard tool to show that a single node obtains slack proportional to its sparsity, see, e.g., \Cref{lem:slackGeneration}. The setting in the following lemma is more complicated, as it does not reason about the slack of a single (sparse) node, but about slack of a collection of (possibly non-sparse) nodes. 

\begin{restatable}{lemma}{lemMatchingSlack}
\label{lem:ordinaryMatchingSlack}
Fix any non-easy clique $C\subseteq V$, let $C\subseteq S\subseteq V$ be a set of nodes, and let $M$ be an arbitrary matching between $C$ and $N(C)\setminus C$.  Then, after $\textsf{SlackGeneration}$ run on $S$, then $C$ contains $\Omega(|M|)$ uncolored nodes with unit-slack in $G[S]$, with probability $1-\exp(-\Omega(|M|))$.
\end{restatable}
\begin{proof}
	Let $M = (c_i,w_i)_i$ be the matching.
	Let $W = \{i : w_i \in S\}$.
	Suppose first that $|W| \ge |M|/2$, when at least half of the nodes $w_i$ are in $S$. 
	Then the set $T = \{i \in W : \text{\begin{math}c_i\end{math} is not activated}\}$ is of size $\Omega(|W|) = \Omega(|M|)$ with probability $\exp(-\Omega(|M|))$ by Chernoff. The nodes $c_i$ with $i \in T$ all obtain unit-slack (by virtue of their neighbor being outside $S$), satisfying the lemma.
	So, assume from now on that $|W| < |M|/2$.
	For simplicity, we redefine $M$ to contain only those pairs $(w_i, c_i)$ where $w_i \not\in S$.
	
	For each $w_i$, let $A_i = C \setminus N(w_i)$ be the set of non-neighbors of $w_i$ in $C$.
	For a color $c \in [\Delta]$, let $Z_c$ be the indicator random variable of the event that some $w_i$ and some $q \in A_i$ tried $c$ and kept it, while $c_i$ was not activated.
	Let $Z = \sum_{c} Z_c$.
	Observe that each color that contributes to $Z$ results in a distinct uncolored node ($c_i$) with unit-slack. Hence, $Z$ lower bounds the number of unit-slack nodes in $C$.
	We first bound the expectation of $Z$ from below and then use Talagrand's inequality to show that it holds with high probability.

	\begin{claim}
		$\Exp[Z] = \Omega(|M|)$.
	\end{claim}
	
	\begin{proof}
		For each $i \in [|M|]$ and each $q \in A_i$, let $Y_{i,q}$ be the indicator random variable of the event that four conditions hold: 
		a) $c_i$ is not activated, b) $w_i$ is activated, trying some color $c \in [\Delta/2]$, c) $q$ also tries $c$, and d) No node in $N_{i,q} = \{w_j\}_j \cup N(w_i) \cup N(q) \setminus \{w_i,q\}$  tries $c$. 
		Let $Y_i = \sum_{q \in A_i} Y_{i,q}$ and let $Y = \sum_i Y_i$.
		Observe that whenever $Y_i = 1$, then $c_i$ receives unit-slack; namely, then $w_i$ tried the same color as some non-neighbor $q \in A_i$, $c_i$ was not activated, and both $w_i$ and $q$ actually kept their color. 
		Thus, $Y$ is a lower bound on the number of unit-slack nodes in $C$ from running \textsf{SlackGeneration}.
		In fact, $Y \le Z$, as it counts how many colors are tried and kept by exactly one pair $w_i, q$, while $Z$ allows multiple pairs to be counted. 
		
		We now bound the expectation of $Y$.
		Note that the four conditions for $Y_{i,q}=1$ are independent. Their probabilities are:
		a) $(1-p_q)$, b) $p_q$, c) $p_g \cdot 2/\Delta$, and d) at least $(1-1/\Delta)^{3\Delta} \ge e^{-12}$ (using that $|N_{i,q}| \le 3\Delta)$.
		Hence, $\Exp[Y_{i,q}] \ge (1-p_q) p_q \cdot 1/\Delta \cdot e^{-6} = \Theta(1/\Delta)$. Thus, $Y = \sum_i \sum_{q \in A_i} \Theta(1/\Delta) = \Omega(|M|)$, using that $|A_i| = \Omega(\Delta)$ (\Cref{lem:acd} (4)).
	\end{proof}
	
	Notice that $Z$ is only a function of the activation and coloring random variables $\omega_u$ of \textsf{SlackGeneration} within $N[C]\cup N(W)$, which are themselves independent.
	We would like to apply Talagrand's inequality (\cref{lem:talagrand}), but unfortunately $Z$ is not $O(1)$-certifiable. Instead, we split it into $Z = T - D$, where
	$T$ counts the number of colors that are tried by some pair $w_i, q$ with $q \in A_i$ where $c_i$ is not activated.
	Also, $D$ counts the number of colors that are tried by some pair $w_i, q$ with $q \in A_i$ and $c_i$ not activated, but where at least one of $w_i$, $q$ did not keep their color.
	We show that $T$ and $D$ that are both Lipschitz and certifiable.
	% Define T and D
	
	\begin{claim}
		$T$ and $D$ are $O(1)$-Lipschitz and $O(1)$-certifiable.
	\end{claim}
	\begin{proof}
		Both of these variables are functions of $\omega_u$, the activation and color trials of nodes $u \in N[C]\cup N(W)$. 
		Both are $2$-Lipschitz: changing the activation/color of a single node can only affect the two colors involved (the one changed to and the one changed from).
		
		Any color that contributes to $T$ is also certifiable by the randomness $\omega_v$ of $3$ nodes: $w_i$, $c_i$, and $q$. $D$ is certifiable by the randomness of $4$ nodes because for one of the trials to be non-successful we require one more node (a $w_j$ node, $j \ne i$, or a neighbor of either $w_i$ or $q$) to try the same color.
	\end{proof}
	
	We now apply Talagrand's inequality (\cref{lem:talagrand}) using these relations to obtain that
	\[
	Pr\left[|T-\E[T]|\ge \E[Z]/10\right]\le \exp\left(-\Theta(1)\frac{(\E[Z]/10-O(\sqrt{\E[T]}))^2}{\E[T]}\right)\le \exp(-\Omega(|M|))\ 
	\]
	where we also used that $\E[T]\le T\le |M| = \Theta(\E[Z])$. Also, with probability $\exp{(-|M|)}$ we obtain that $|T-\E[T]|\le \E[Z]/10$. Similarly, with the same probability we obtain $|D-\E[D]|<\E[Z]/10$, since $D \le T$. By the linearity of expectation, $\E[T]-\E[D] = \E[Z]$.
	% Completing the proof
	Putting together we see that, with probability $2\exp{(-\Omega(|M|))}=\exp{(-\Omega(|M|))}$, 
	\[ |Z - \Exp[Z]| = |(T-\E[T])-(D-\E[D])| \le 2\E[Z]/10\ . \]
	Thus, with probability $\exp{(-\Omega(|M|))}$, $Z \ge 4 \E[Z]/5 = \Omega(|M|)$, which completes the proof.
\end{proof}

Now, we are ready to prove that ordinary ACs obtain many uncolored nodes with unit-slack.
\begin{lemma}[Slack in Ordinary ACs]	Fix an ordinary AC $C$, let $C\subseteq S\subseteq V$ be a set of nodes,
\label{lem:ordinarySlack}
After $\textsf{SlackGeneration}(S)$ each ordinary AC $C$ contains at least $\Omega(\psi)$ uncolored nodes with unit-slack in $G[S]$, with probability $1-\exp(-\Omega(\psi))$.
\end{lemma}
\begin{proof}
\textit{Case $|C|\leq \Delta-\psi$.}
Let $v$ be an arbitrary node of $C$. As the clique is non-nice we have $deg_C(v)=|C|$. We obtain that the palette size of $v$ exceeds its degree by $\psi-e(v)$, because $\Delta-deg(v)=\Delta-(e(v)+deg_C(v))=\Delta-(e(v)+|C|)\geq \Delta-e(v)-(\Delta-\psi)=\psi-e(v)$. 
Thus, if, $e(v)\leq \psi/2$,  node $v$ has permanent slack (a priori) of at least $\psi/2$, regardless of the outcome of the random bits of \textsf{SlackGeneration}. Now, consider the case that $e(v)\geq \psi/2$. Then, \Cref{cor:externalDegreeSlack} provides $\Omega(e(v))=\Omega(\psi)$ slack for $v$ with probability $1-\exp(-\Omega(e(v))=1-\exp(-\Omega(\psi))$, which also holds in the stated subgraph. Additionally, each node $v$ remains uncolored with probability $\geq 1-p_g\geq 1/2$, and a Chernoff bound shows that $\Omega(\Delta)=\Omega(\psi)$ nodes of the clique remain uncolored with error probability $\exp(-\Omega(\Delta))$.  A union bound over these events shows the claim. 

\textit{Case $|C|> \Delta-\psi$:} By \Cref{obs:vertexCoverMatching} there exists a matching $M = (c_i,w_i)_i$ of size $\Omega(\psi)$ between $C$  and $N(C)\setminus C$ and by the lemma assumption  all nodes of $C$ participate in \textsf{SlackGeneration}. Hence, we can apply \Cref{lem:ordinaryMatchingSlack} to obtain that the ordinary clique obtains $\Omega(|M|)=\Omega(\psi)$ nodes with unit-slack with probability $1-\exp{(\Omega(-\psi))}$. 
\end{proof}

Escape nodes also get slack via the many non-edges between their neighbors.

\begin{lemma}
\label{lem:escapeSlack}
Let $S\supseteq\cR$ be a subset of nodes. After  $\textsf{SlackGeneration}(S)$, each escape node has slack $\Omega(\psi)$ in $G[V]$, with probability $1-\exp(-\Omega(\psi))$.
\end{lemma}
\begin{proof}
Since a runaway AC $C$ is a clique, each node in $C$ has at least $\Delta - \psi$ neighbors inside $C$, and therefore an outside degree at most $\psi$.
Let $v$ be an escape node.
Let $C_1$ and $C_2$ be two ACs that $v$ is an escape for, and
let $X_i, i=1,2$ be the set of neighbors of $v$ in $C_i$. By the escape definition of $v$, we have $|X_i|\geq \phi$ and as we observed the outside degree of each node in $X_1$ is at most $\psi$. Hence, the bipartite graph between $X_1$ and $X_2$ is sparse, i.e.,  contains at least $|X_1|\cdot |X_2|-|X_1|\cdot\psi\geq |X_1|\cdot (\phi-\psi)=\Omega(\phi^2)$ non-edges. 
It follows that the local sparsity of $v$ in $G[\cR]$ (and also in $G[V]$) is at least $\zeta_v=\Omega(\phi^2/\Delta)=\Omega(\psi)$.  All nodes of $X_1\cup X_2\subseteq\cR\subseteq S$ participate in \textsf{SlackGeneration}, so via \Cref{lem:slackGeneration} the escape node $v$ obtains slack at least $\Omega(\zeta)=\Omega(\psi)$ with probability $1-\exp(-\Omega(\psi))$ in $G[\cR]$. The slack is permanent and also holds in $G[V]$.
\end{proof}

\begin{proof}[Proof of \Cref{lem:slackGenerationProperties}]
Property $(i)$ follows with \Cref{lem:slackGeneration} for a sparse node (as it has sparsity $\eps\Delta=\Omega(\log^3 n)$) and with \Cref{lem:escapeSlack} for an escape node. Property $(ii)$ holds for a single ordinary AC with \Cref{lem:ordinarySlack}, as all nodes of the ordinary clique participate in $\slackgeneration$. As these claims hold w.h.p. the claim follows for all sparse nodes, all escape nodes and all ordinary cliques via a union bound. 
Property $(iii)$ follows as each node that participates in \textsf{SlackGeneration} remains uncolored with probability $\geq 1-p_g\geq 3/4$. Thus, for each of special node of a difficult AC, in expectation, at least $3/4$ of the respective node set remains uncolored. Note that a difficult AC can neither contain a protector nor an escape, see \Cref{obs:difficultNoSpecial}. A Chernoff bound shows that each condition is satisfied with probability $\geq 1-1/n^3$, and a union bound over the conditions (as there are only $O(n/\Delta)$ difficult ACs) proves the claim. 
\end{proof}

\section{\texorpdfstring{$\Delta$-coloring in the Congested Clique}{Delta-coloring in the Congested Clique}}\label{sec:CC}
In the Congested Clique model, each node can send messages of size $O(\log n)$ to each of the other $n-1$ nodes in every round. The goal is to minimize the total number of rounds needed for each node to come up with its part of the solution. 

In this section we prove the following theorem.

\thmMaincongestedClique*

The proof uses two components. We first extend a Congested Clique coloring algorithm of \cite{CFGUZ18} to a more flexible range of palettes (\cref{lemma:d1LC}). We then argue (\cref{lem:balanced}) that all but one of the $deg+1$-list coloring instances produced by our algorithm of \Cref{sec:algHighDegree} satisfies the prerequisites of \cref{lemma:d1LC}. The remaining instance has small size and can therefore be solved in $O(1)$ rounds of Congested Clique by sending it to one node via Lenzen's routing \cite{lenzen13}.

\begin{lemma}[Theorem 3.2 in \cite{CFGUZ18}]\label{lemma:d1LC}%
There is a $O(1)$-round Congested Clique algorithm to solve any list coloring instance in which each node $v$ has a list of size at least $\max(d_v+1, c\cdot \Delta, \log^{4+\eps} n)$, for some constants $c,\eps > 0$. 
\end{lemma}
\begin{proof}
The proof is similar to the proof of \cite[Theorem 3.2]{CFGUZ18}. 
The main change is the generalization of the graph partitioning \cite[Lemma 3.1]{CFGUZ18} lemma to $\Delta'=c\cdot \Delta$, where the assumed palette size, aka list size, changes to $|\Psi(v)|\geq \max\{deg_G(v),c\cdot\Delta\}$ instead of $|\Psi(v)|\geq \max\{deg_G(v),\Delta-\Delta^{\lambda}\}+1$.

For the sake of readability, we repeat the properties of \cite[Lemma 3.1]{CFGUZ18} here, in a slightly modified way to suit our setting.
We partition the vertices into sets $B_1, \dotsc, B_k$, and a leftover set $L$, where $k=\sqrt{\Delta}$. Namely, we assign each node to $L$ with probability $q = \Theta(\sqrt{\log n}/\Delta^{1/4})$ and assign other nodes u.a.r.\ to the other sets. We also partition the colors u.a.r.\ into sets $C_1, \dotsc, C_k$, i.e., assign each color to set $C_i$ with probability $p = 1/k$. 
The claim of Lemma 3.1 in \cite{CFGUZ18} is that these partitions can be computed in $O(1)$ rounds of the Congested Clique, so that the following properties hold with probability $1-1/\poly(n)$:
\begin{description}
\item[i) Size of each part:] We have $|E(G[B_i])| = O(|V|)$, for each $i \in [k]$, and $|L| = O(q |V|) = O(\frac{\sqrt{\log n}}{\Delta^{1/4}}) \cdot |V|$. 

\item[ii) Available colors in $B_i$:] For each $i\in \{1, \ldots,k\}$ and $v \in B_i$, the number of available colors for $v$ in the subgraph $B_i$ is $g_i(v) := |\Psi(v) \cap C_i|$. We have $g_i(v) \geq \max\{\deg_{B_i}(v), c\cdot \Delta_i\}+1$, where 
$\Delta_i := \max_{v \in B_i}\deg_{B_i}(v)$. 

\item[iii) Available colors in $L$:] For each $v \in L$, define $g_L(v) := |\Psi(v)| - (\deg_G(v) - \deg_L(v))$. We have $g_L(v) \geq \max\{\deg_L(v), c\cdot \Delta_L\}+1$ for each $v \in L$, where $\Delta_L := \max_{v \in L}\deg_{L}(v)$. Note that $g_L(v)$ represents a lower bound on the number of available colors in $v$ {\em after} all of $B_1, \ldots, B_k$ have been colored.

\item[iv) Remaining degrees:] The maximum degrees of $B_i$ and $L$ are $\max_{v \in B_i} \deg_{B_i}(v) = O(\sqrt{\Delta})$ and $\max_{v \in L} \deg_{L}(v) = O(q\Delta) =  O(\frac{\sqrt{\log n}}{\Delta^{1/4}}) \cdot \Delta$. For each vertex $v$, we have $\deg_{B_i}(v)\leq \max\{O(\log n), O(1/\sqrt{\Delta}) \cdot \deg_G(v)\}$ if $v \in B_i$, and $\deg_{L}(v)\leq \max\{O(\log n), O(q) \cdot \deg_G(v)\}$ if $v \in L$.
\end{description}

Our strengthened restatement of Lemma 3.1 of \cite{CFGUZ18} then becomes:

\begin{lemma}[Theorem 3.2 in \cite{CFGUZ18}]\label{T:partition}
Suppose $|\Psi(v)| \ge \max(\deg_G(v), c \Delta)$ and $|V| > \Delta = \omega(\log^2 n)$.
The two partitions $V = B_1 \cup \cdots B_k \cup L$ and $C = \bigcup_{v \in V} \Psi(v) = C_1 \cup \cdots \cup C_k$ satisfy the required properties, w.h.p.
\end{lemma}

We illustrate here how to modify the proof for the properties of $B_i$ (Property ii) and iv)). The ones for $L$ are similar. Property i) and the second half of Property iv) follow unchanged from \cite{CFGUZ18}.

Let $d_i(v) = \deg_{B_i}(v)$ and let $\Delta_i = \max_{v\in B_i} d_i(v)$.
Let $d'(v) = \max(\deg_G(v), c\Delta)$ and let $d'_i(v) = \max(\deg_{B_i}(v),c \Delta_i)$.
Let $\epsilon_1 = q/3$.

\paragraph{Property iv) ($B_i$):}
The expected value of $d_i(v)$ is $(1-q)p \deg_G(v)$, so by Chernoff, it holds that
$d_i(v) \le (1+\epsilon_1) (1-q)p \deg_G(v)$ with probability $1-\exp(-\Theta(\epsilon_1^2) p \deg_G(v))$. 
In particular, $d'_i(v) \le (1+\epsilon_1) (1-q)p d'(v)$ with probability $1-\exp(-\Theta(\epsilon_1^2) p d'(v)) = 1-1/\poly(n)$. 
Thus, $\Delta_i \le (1+\epsilon_1) (1-q)p \Delta \le (1-2q/3)p \Delta$, w.h.p.
Note that since $d'(v) \ge c \Delta$, we have that
$d'_i(v) \ge c p \Delta = c \sqrt{\Delta}$.

\paragraph{Property ii):}
We show that with probability $1-1/\poly(n)$, we have $|\Psi(v)\cap C_i|\geq d'_i(v)+1$. 
The palette of $v$ in $G$ is at least $|\Psi(v)| \ge d'(v)+1$.
The expected size of $|\Psi(v)\cap C_i|$ is $p |\Psi(v)| \ge p(d'(v)+1)$.
Thus by Chernoff, $|\Psi(v) \cap C_i| \ge (1-\epsilon_1) p |\psi(v)|$ with probability $1- \exp(-\Theta(\epsilon_1^2) d'(v)) = 1- \exp(-\Theta(\epsilon_1^2) c \sqrt{\Delta}) = 1-1/\poly(n)$. 
Hence, we have that
\[ |\Psi(v) \cap C_i| \ge (1-q/3) p |\psi(v)| 
\ge (1-q/3) p d'(v) \ge (1+q/3) d'_i(v)\ . \]
To complete our claim, we observe that $d'_i(v) \cdot q/3 = \Theta(\Delta^{1/4} \cdot \sqrt{\log n})$ is at least 1.

The rest of the proof of the adapted version of \cite[Lemma 3.1]{CFGUZ18} and the adapted version of \cite[Theorem 3.2]{CFGUZ18}, hence of our lemma, goes through the same. 
\hfill
\end{proof}

We now argue that the $deg+1$-list coloring instances produced by \Cref{thm:ListColorReduction} have a nice structure.
We say that a $deg+1$-list coloring instance of maximum degree $\Delta_0$ is \emph{$c$-balanced} if each node has palette of size $\max(\deg_G(v)+1,\Delta_0/c)$, and it is \emph{balanced} if it is $c$-balanced for some absolute constant $c> 0$.

Let $P$ be the set of pairs $(p_C,w_C)$ for guarded ACs $C$.
Let $H$ be the (virtual) graph formed on $P$, with an edge between two pairs if they contain adjacent nodes (see \cref{lem:guardedPair-coloring}).

\begin{lemma}
\label{lem:balanced}
All the $deg+1$-list coloring instances produced by our main algorithm (Alg.\ 1) are balanced, except for the graph $H$.
The graph $H$ contains $O(n/\Delta)$ nodes.
\end{lemma}

\begin{proof} 
All but four cases have already been shown to involve palettes of size $\Omega(\Delta)$; see the accounting in the subsections of \cref{ssec:detailedSteps}.
These are balanced by definition.
We now address the remaining cases.

\emph{Runaway ACs (white instance):} These are the neighbors of the escape. Denote their number by $t$ and note that by \cref{lem:runawaylists}, $t \ge \phi/2$.
The outside degree of these nodes is at most $\psi$ (as difficult ACs have at least $\Delta-\psi$ nodes and no non-edges).
Thus, the nodes have degrees between $\phi/2-1$ and $\phi/2+\psi \le \phi$. Hence, the instance is balanced.

\emph{Guarded ACs (white instance):} This is the set $S$ of common neighbors of the pair $(u_C,p_C)$ for the guarded AC $C$.
By the definition of protectors, $|S| \ge \phi$.
By the same reasoning as for the runaway ACs, the degrees in $S$ are between $\phi-1$ and $2\phi$.

\emph{Guarded ACs (pairs):}
The set $P$ of pairs $(u_C,p_C)$ contains at most $2n/((1-\epsilon)\Delta) \le 3n/\Delta$ nodes, since there is only pair for each guarded AC, and there are at most $n/((1-\epsilon)\Delta)$ ACs.
Hence, $H$ has $O(n/\Delta)$ nodes.

\emph{Ordinary ACs (white instance):} 
If an AC has size at most $\Delta - \psi$, then all its nodes get slack and are white. As $\Omega(\Delta)$ nodes are uncolored, the instance is balanced.

In other ordinary ACs, there are $\Omega(\psi)$ white nodes (of unit slack) \cref{lem:ordinarySlack}. They have external degree at most $\psi$, since the clique has size at least $\Delta - \psi$.
Thus, their degrees and palettes are in the range $\Omega(\psi)$ to $\psi$, and thus the instance is balanced.
\hfill
\end{proof}

% For the low-degree case, the issue is generating slack; otherwise, the same approach works.
% Various approaches are possible here. It could be formulated as an instance of the Lovasz Local Lemma; some form of opportunistic information gather might be possible.

\begin{proof} \emph{(Proof of \cref{cor:CCmain}).}
We implement Alg.~1 in Congested Clique.
Slack generation takes 0 rounds and ACD construction $O(1)$ rounds via \cref{lem:acd}. 
By \cref{lem:balanced}, the $deg+1$-list coloring instances generated by the algorithm are all balanced, except for the graph $H$. 
The balanced instances can be solved via \cref{lemma:d1LC}.
The graph $H$ is on $n/\Delta$ nodes with degree at most $\Delta$, so it is of total size $O(n)$. It can then be sent to a single node (by Lenzen's transform) and solved there centrally.
\hfill
\end{proof}

\clearpage

\appendix
\addtocontents{toc}{\protect\setcounter{tocdepth}{0}}
\section{Additional Figure(s)}
\iffalse
\begin{figure}[h]
\centering
	\includegraphics[width=0.65\textwidth, page=1]{slack.pdf}
	\caption{The three types of slack. An uncolored node can obtain permanent slack by two neighbors receiving the same color (left); a priori slack by having a degree $<\Delta$ (middle); or temporary slack by stalling some of its neighbors, to be colored later (right).}
	\label[figure]{fig:slack-types}
\end{figure}
\fi
\begin{figure}[h]
	\centering
	\includegraphics[width=0.7\textwidth, page=1]{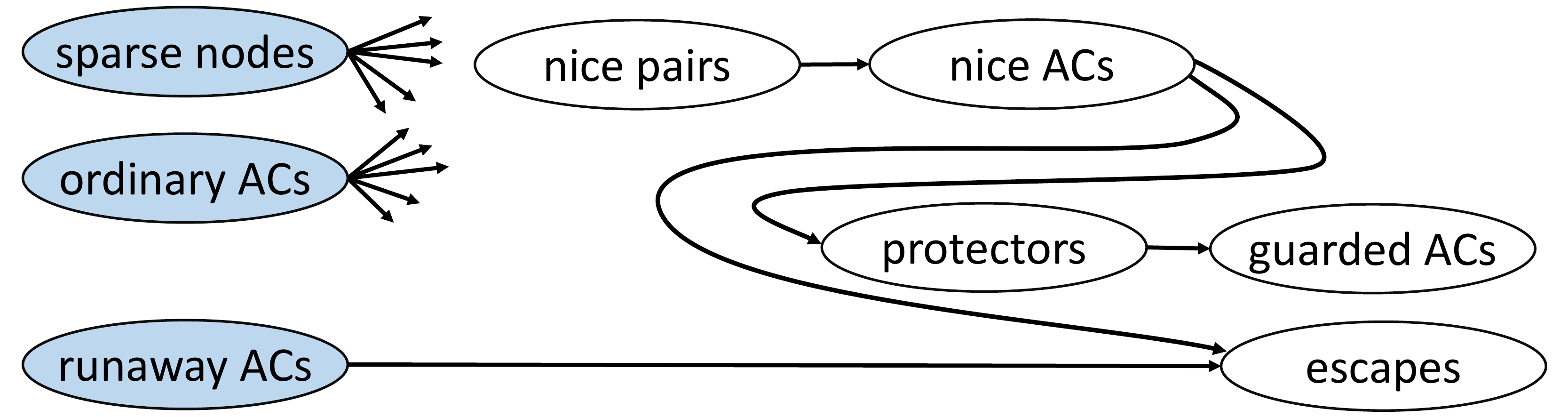}
	\caption{The order in which nodes need to be colored. Blue sets of nodes participate in \textsf{SlackGeneration}. An arrow indicates, that the respective set of nodes has to be colored before the latter one. Sparse nodes and ordinary cliques need to be colored before  everyone, except for nodes in runaway ACs.
	}
	\label[figure]{figure:ACorder}
\end{figure}

\section{Concentration Inequalities}
We will also use the following variant of Talagrand's inequality. A function $f(x_1, \dots, x_n)$ is called \emph{$c$-Lipschitz} iff changing any single $x_i$ can affect the value of
$f$ by at most $c$. Additionally, $f$ is called \emph{$r$-certifiable} iff whenever $f(x_1, \dots , x_n) \ge s$, there exists at
most $r s$ variables $x_{i_1}, \dots, x_{i_{rs}}$
so that knowing the values of these variables certifies $f\ge s$.
\begin{lemma}[Talagrand's Inequality \cite{molloy2013coloring}]
	\label{lem:talagrand}
	Let $X_1, \dots , X_n$ be $n$ independent random variables and $f(X_1, \dots , X_n)$ be a $c$-Lipschitz $r$-certifiable function. %For any $t \ge 1$,
	For any $b \ge 1$,
	\[
	P (|f - \E[f]| > b+60c\sqrt{r\E[f]}) \le 4 \exp\left(-\frac{b^2}{8c^2r\E[f]}\right)\ .
	\]
\end{lemma}

\bibliographystyle{alpha}
\bibliography{references}

\newcommand{\etalchar}[1]{$^{#1}$}
\begin{thebibliography}{GHKM21}

\bibitem[ABI86]{alon86}
Noga Alon, L\'aszl\'o Babai, and Alon Itai.
\newblock A fast and simple randomized parallel algorithm for the maximal
  independent set problem.
\newblock {\em J.\ of Algorithms}, 7(4):567--583, 1986.

\bibitem[ACK19]{ACK19}
Sepehr Assadi, Yu~Chen, and Sanjeev Khanna.
\newblock {Sublinear algorithms for {$(\Delta + 1)$} vertex coloring}.
\newblock In {\em the Proceedings of the ACM-SIAM Symposium on Discrete
  Algorithms (SODA)}, pages 767--786, 2019.
\newblock Full version at arXiv:1807.08886.

\bibitem[AKM22]{AKM22}
Sepehr Assadi, Pankaj Kumar, and Parth Mittal.
\newblock Brooks’ theorem in graph streams: a single-pass semi-streaming
  algorithm for {$\Delta$}-coloring.
\newblock In {\em Proceedings of the 54th Annual ACM SIGACT Symposium on Theory
  of Computing}, pages 234--247, 2022.

\bibitem[Bar16]{Barenboim16}
Leonid Barenboim.
\newblock Deterministic ({\(\Delta\)} + 1)-coloring in sublinear (in
  {\(\Delta\)}) time in static, dynamic, and faulty networks.
\newblock {\em Journal of the ACM}, 63(5):47:1--47:22, 2016.

\bibitem[BBH{\etalchar{+}}19]{Balliu2019}
Alkida Balliu, Sebastian Brandt, Juho Hirvonen, Dennis Olivetti, Mika{\"{e}}l
  Rabie, and Jukka Suomela.
\newblock Lower bounds for maximal matchings and maximal independent sets.
\newblock In {\em Proc.\ 60th {IEEE} Symp.\ on Foundations of Computer Science
  (FOCS)}, pages 481--497, 2019.

\bibitem[BBKO21]{BBKOmis}
Alkida Balliu, Sebastian Brandt, Fabian Kuhn, and Dennis Olivetti.
\newblock Improved distributed lower bounds for {MIS} and bounded (out-)degree
  dominating sets in trees.
\newblock In {\em Proc.\ 40th ACM Symposium on Principles of Distributed
  Computing (PODC)}, 2021.

\bibitem[BBKO22]{BBKO22}
Alkida {Balliu}, Sebastian {Brandt}, Fabian {Kuhn}, and Dennis {Olivetti}.
\newblock Distributed {$\Delta$}-coloring plays hide-and-seek.
\newblock In {\em {STOC} '22: 54rd Annual {ACM} {SIGACT} Symposium on Theory of
  Computing}. {ACM}, 2022.
\newblock Full version at CoRR abs/2110.00643.

\bibitem[BBO20]{BBORules}
Alkida Balliu, Sebastian Brandt, and Dennis Olivetti.
\newblock Distributed lower bounds for ruling sets.
\newblock In {\em Proc.\ 61st {IEEE} Symp.\ on Foundations of Computer Science
  (FOCS)}, pages 365--376, 2020.

\bibitem[BE13]{barenboimelkin_book}
Leonid Barenboim and Michael Elkin.
\newblock {\em Distributed Graph Coloring: Fundamentals and Recent
  Developments}.
\newblock Morgan \& Claypool Publishers, 2013.

\bibitem[BEG18]{BEG18}
Leonid Barenboim, Michael Elkin, and Uri Goldenberg.
\newblock {Locally-Iterative Distributed ({\(\Delta\)}+ 1)-Coloring below
  Szegedy-Vishwanathan Barrier, and Applications to Self-Stabilization and to
  Restricted-Bandwidth Models}.
\newblock In {\em the Proceedings of the ACM Symposium on Principles of
  Distributed Computing (PODC)}, pages 437--446, 2018.

\bibitem[BEPS16]{BEPSv3}
Leonid Barenboim, Michael Elkin, Seth Pettie, and Johannes Schneider.
\newblock The locality of distributed symmetry breaking.
\newblock {\em Journal of the ACM}, 63(3):20:1--20:45, 2016.

\bibitem[BFH{\etalchar{+}}16]{brandt2016LLL}
Sebastian Brandt, Orr Fischer, Juho Hirvonen, Barbara Keller, Tuomo
  Lempi{\"a}inen, Joel Rybicki, Jukka Suomela, and Jara Uitto.
\newblock A lower bound for the distributed lov{\'a}sz local lemma.
\newblock In {\em Proc. 48th ACM Symposium on Theory of Computing (STOC 2016)},
  pages 479--488. ACM, 2016.

\bibitem[BKM20]{BKM20}
Philipp Bamberger, Fabian Kuhn, and Yannic Maus.
\newblock Efficient deterministic distributed coloring with small bandwidth.
\newblock In Yuval Emek and Christian Cachin, editors, {\em {PODC} '20: {ACM}
  Symposium on Principles of Distributed Computing, Virtual Event, Italy,
  August 3-7, 2020}, pages 243--252. {ACM}, 2020.

\bibitem[BO20]{BrandtO20}
Sebastian Brandt and Dennis Olivetti.
\newblock Truly tight-in-{\(\Delta\)} bounds for bipartite maximal matching and
  variants.
\newblock In {\em Proceedings of the 2020 {ACM} Symposium on Principles of
  Distributed Computing, (PODC)}, pages 69--78, 2020.

\bibitem[Bra19]{Brandt19}
Sebastian Brandt.
\newblock An automatic speedup theorem for distributed problems.
\newblock In Peter Robinson and Faith Ellen, editors, {\em Proceedings of the
  2019 {ACM} Symposium on Principles of Distributed Computing, {PODC} 2019,
  Toronto, ON, Canada, July 29 - August 2, 2019}, pages 379--388. {ACM}, 2019.

\bibitem[Bro41]{brooks_1941}
R.~Leonard Brooks.
\newblock On colouring the nodes of a network.
\newblock {\em Mathematical Proceedings of the Cambridge Philosophical
  Society}, 37(2):194–197, 1941.

\bibitem[CDP21]{CDP21}
Artur Czumaj, Peter Davies, and Merav Parter.
\newblock Simple, deterministic, constant-round coloring in congested clique
  and {MPC}.
\newblock {\em {SIAM} J. Comput.}, 50(5):1603--1626, 2021.

\bibitem[CFG{\etalchar{+}}18]{CFGUZ18}
Yi-Jun Chang, Manuela Fischer, Mohsen Ghaffari, Jara Uitto, and Yufan Zheng.
\newblock The complexity of {($\Delta+ 1$)} coloring in congested clique,
  massively parallel computation, and centralized local computation.
\newblock {\em arXiv preprint arXiv:1808.08419}, 2018.
\newblock Extended abstract appears in PODC'19.

\bibitem[CKP16]{chang2016exponential}
Yi-Jun Chang, Tsvi Kopelowitz, and Seth Pettie.
\newblock An exponential separation between randomized and deterministic
  complexity in the local model.
\newblock In {\em Proc. 57th Symposium on Foundations of Computer Science (FOCS
  2016)}, pages 615--624. IEEE, 2016.

\bibitem[CLP20]{CLP20}
Yi-Jun Chang, Wenzheng Li, and Seth Pettie.
\newblock Distributed {($\Delta+1$)}-coloring via ultrafast graph shattering.
\newblock {\em SIAM Journal of Computing}, 49(3):497--539, 2020.

\bibitem[CP17]{ChangHierarchy17}
Yi-Jun Chang and Seth Pettie.
\newblock {A Time Hierarchy Theorem for the LOCAL Model}.
\newblock In {\em the Proceedings of the Symposium on Foundations of Computer
  Science (FOCS)}, pages 156--167, 2017.

\bibitem[EPS15]{EPS15}
Michael Elkin, Seth Pettie, and Hsin{-}Hao Su.
\newblock (2{\(\Delta-1\)})-edge-coloring is much easier than maximal matching
  in the distributed setting.
\newblock In {\em Proceedings of the Twenty-Sixth Annual {ACM-SIAM} Symposium
  on Discrete Algorithms, {SODA} 2015, San Diego, CA, USA, January 4-6, 2015},
  pages 355--370, 2015.

\bibitem[FHK16]{FHK}
Pierre Fraigniaud, Marc Heinrich, and Adrian Kosowski.
\newblock Local conflict coloring.
\newblock In {\em the Proceedings of the Symposium on Foundations of Computer
  Science (FOCS)}, pages 625--634, 2016.

\bibitem[GHKM21]{GHKM21}
Mohsen Ghaffari, Juho Hirvonen, Fabian Kuhn, and Yannic Maus.
\newblock Improved distributed delta-coloring.
\newblock {\em Distributed Comput.}, 34(4):239--258, 2021.

\bibitem[GK21]{GK21}
Mohsen Ghaffari and Fabian Kuhn.
\newblock Deterministic distributed vertex coloring: Simpler, faster, and
  without network decomposition.
\newblock In {\em FOCS}, 2021.

\bibitem[HKMT21]{HKMT21}
Magn{\'{u}}s~M. Halld{\'{o}}rsson, Fabian Kuhn, Yannic Maus, and Tigran
  Tonoyan.
\newblock Efficient randomized distributed coloring in {CONGEST}.
\newblock In {\em {STOC} '21: 53rd Annual {ACM} {SIGACT} Symposium on Theory of
  Computing, Virtual Event, Italy, June 21-25, 2021}, pages 1180--1193. {ACM},
  2021.
\newblock Full version at CoRR abs/2105.04700.

\bibitem[HKNT22]{HKNT22}
Magn{\'{u}}s~M. Halld{\'{o}}rsson, Fabian Kuhn, Alexandre Nolin, and Tigran
  Tonoyan.
\newblock Near-optimal distributed degree+1 coloring.
\newblock In {\em STOC}, 2022.

\bibitem[HMKS16]{disc16_coloring}
Dan Hefetz, Yannic Maus, Fabian Kuhn, and Angelika Steger.
\newblock A polynomial lower bound for distributed graph coloring in a weak
  {LOCAL} model.
\newblock In {\em Proc.\ 30th Symp.\ on Distributed Computing (DISC)}, pages
  99--113, 2016.

\bibitem[HNT22]{HNT22}
Magn{\'{u}}s~M. Halld{\'{o}}rsson, Alexandre Nolin, and Tigran Tonoyan.
\newblock Overcoming congestion in distributed coloring.
\newblock In {\em PODC}, 2022.

\bibitem[HSS18]{HSS18}
David~G. Harris, Johannes Schneider, and Hsin-Hao Su.
\newblock {Distributed {($\Delta + 1$)}-coloring in sublogarithmic rounds}.
\newblock {\em Journal of the ACM}, 65:19:1--19:21, 2018.

\bibitem[Joh99]{johansson99}
{\"{O}}jvind Johansson.
\newblock Simple distributed \emph{Delta}+1-coloring of graphs.
\newblock {\em Inf. Process. Lett.}, 70(5):229--232, 1999.

\bibitem[Kuh20]{KuhnRecursive20}
Fabian Kuhn.
\newblock Faster deterministic distributed coloring through recursive list
  coloring.
\newblock In Shuchi Chawla, editor, {\em Proceedings of the 2020 {ACM-SIAM}
  Symposium on Discrete Algorithms, {SODA} 2020, Salt Lake City, UT, USA,
  January 5-8, 2020}, pages 1244--1259. {SIAM}, 2020.

\bibitem[Len13]{lenzen13}
Christoph Lenzen.
\newblock Optimal deterministic routing and sorting on the congested clique.
\newblock In {\em the Proceedings of the ACM Symposium on Principles of
  Distributed Computing (PODC)}, pages 42--50. {ACM}, 2013.

\bibitem[Lin92]{linial92}
Nati Linial.
\newblock Locality in distributed graph algorithms.
\newblock {\em SIAM Journal on Computing}, 21(1):193--201, 1992.

\bibitem[Lub86]{luby86}
Michael Luby.
\newblock A simple parallel algorithm for the maximal independent set problem.
\newblock {\em SIAM Journal on Computing}, 15:1036--1053, 1986.

\bibitem[Mau21]{M21}
Yannic Maus.
\newblock Distributed graph coloring made easy.
\newblock In Kunal Agrawal and Yossi Azar, editors, {\em {SPAA} '21: 33rd {ACM}
  Symposium on Parallelism in Algorithms and Architectures, Virtual Event, USA,
  6-8 July, 2021}, pages 362--372. {ACM}, 2021.

\bibitem[MR98]{Molloy1998}
Michael Molloy and Bruce Reed.
\newblock {Further Algorithmic Aspects of the Local Lemma}.
\newblock In {\em the Proceedings of the ACM-SIAM Symposium on Discrete
  Algorithms (SODA)}, pages 524--529, 1998.

\bibitem[MR13]{molloy2013coloring}
Michael Molloy and Bruce Reed.
\newblock {\em Graph colouring and the probabilistic method}, volume~23.
\newblock Springer Science \& Business Media, 2013.

\bibitem[MT20]{MT20}
Yannic Maus and Tigran Tonoyan.
\newblock Local conflict coloring revisited: Linial for lists.
\newblock In {\em 34th International Symposium on Distributed Computing, {DISC}
  2020, October 12-16, 2020, Virtual Conference}, pages 16:1--16:18, 2020.

\bibitem[MU21]{MU21}
Yannic Maus and Jara Uitto.
\newblock Efficient {CONGEST} algorithms for the lov{\'{a}}sz local lemma.
\newblock In Seth Gilbert, editor, {\em 35th International Symposium on
  Distributed Computing, {DISC} 2021, October 4-8, 2021, Freiburg, Germany
  (Virtual Conference)}, volume 209 of {\em LIPIcs}, pages 31:1--31:19. Schloss
  Dagstuhl - Leibniz-Zentrum f{\"{u}}r Informatik, 2021.

\bibitem[NS95]{naor95}
Moni Naor and Larry Stockmeyer.
\newblock What can be computed locally?
\newblock {\em SIAM J.\ on Comp.}, 24(6):1259--1277, 1995.

\bibitem[PS95a]{PS95}
Alessandro Panconesi and Aravind Srinivasan.
\newblock The local nature of {$\Delta$}-coloring and its algorithmic
  applications.
\newblock {\em Combinatorica}, 15(2):255--280, 1995.

\bibitem[PS95b]{panconesi95}
Alessandro Panconesi and Aravind Srinivasan.
\newblock On the complexity of distributed network decomposition.
\newblock {\em Journal of Algorithms}, 20(2):581--592, 1995.

\bibitem[RG20]{RG20}
V{\'{a}}clav Rozho\v{n} and Mohsen Ghaffari.
\newblock Polylogarithmic-time deterministic network decomposition and
  distributed derandomization.
\newblock In {\em the Proceedings of the ACM-SIAM Symposium on Discrete
  Algorithms (SODA)}, pages 350--363, 2020.

\bibitem[SW10]{SW10}
Johannes Schneider and Roger Wattenhofer.
\newblock A new technique for distributed symmetry breaking.
\newblock In {\em the Proceedings of the ACM Symposium on Principles of
  Distributed Computing (PODC)}, pages 257--266. {ACM}, 2010.

\end{thebibliography}
\end{document}